\pdfoutput=1
\documentclass{article}
\usepackage{fullpage}
\usepackage{hyperref}
\hypersetup{colorlinks=true, allcolors=blue}
\usepackage{amsmath,amsfonts,amsthm,amssymb}
\usepackage{color}
\usepackage{graphicx,float,wrapfig}
\usepackage[toc,page]{appendix}
\usepackage{natbib}
\usepackage{thmtools,thm-restate}
\usepackage{multirow}
\usepackage{multicol}
\usepackage{array}
\usepackage[table]{xcolor}
\usepackage{nicematrix}

\usepackage{xcolor}
\usepackage[algo2e,ruled,linesnumbered,boxed,noend]{algorithm2e}
\usepackage{verbatim}
\usepackage{graphicx,amsfonts,amsmath,amssymb,epsfig,color,multicol,pstricks,tikz, pgf}
\usetikzlibrary{decorations.markings}
\usetikzlibrary{patterns}
\usetikzlibrary{calc, intersections}
\usepackage{pgf,tikz}
\usetikzlibrary{calc}
\usepackage{tkz-euclide}
\usepackage{pgfplots}
\usepackage{mathabx}
\usepackage{graphicx,amsfonts,amsmath,amssymb,epsfig,color,multicol,pstricks,tikz, pgf}

\def\notes{0}
\ifnum\notes=1
\newcommand{\xhk}[1]{\textcolor{red}{[xhk: #1]}}

\newcommand{\Piotr}[1]{\textcolor{blue}{Piotr: #1}}
\newcommand{\tnote}[1]{\textcolor{magenta}{Talya: #1}}
\newcommand{\tnew}[1]{\textcolor{blue}{#1}}
\else
\newcommand{\xhk}[1]{}

\newcommand{\Piotr}[1]{}
\newcommand{\tnote}[1]{}
\newcommand{\tnew}[1]{#1}

\fi

\newcommand{\mb}[1]{\mathbb{#1}}

\newcommand{\algoname}{A* search}

\DeclareMathOperator*{\argmin}{argmin}

\newtheorem{theorem}{Theorem}[section]
\newtheorem{lemma}[theorem]{Lemma}
\newtheorem{proposition}[theorem]{Proposition}

\newtheorem{definition}[theorem]{Definition}
\newtheorem{assumption}[theorem]{Assumption}
\newtheorem{fact}[theorem]{Fact}

\numberwithin{equation}{section}

\title{Embeddings and labeling schemes for A*}

    \author{Talya Eden  \thanks{CSAIL at MIT,  Boston University Department of Computer Science, \textit{talyaa01@gmail.com}. Partially supported by the NSF Grant CCF-1740751, the Eric and
        Wendy Schmidt Fund,  Ben-Gurion University, and the Computer Science Department at Boston University.} \\ Boston University and MIT
    \and Piotr Indyk \thanks{CSAIL at MIT, \textit{indyk@mit.edu}. Partially supported by the NSF TRIPODS program (awards CCF-1740751 and DMS-2022448) and Simons Investigator Award.}\\ MIT 
    \and    Haike Xu \thanks{IIIS, Tsinghua University, \textit{xhk18@mails.tsinghua.edu.cn}.} \\ Tsinghua University
}
\date{\vspace{-5ex}}

\begin{document}
\maketitle

\usetikzlibrary{trees}

	\tikzset{dotted pattern/.style args={#1}{
		postaction=decorate,
		decoration={
			markings,
			mark=between positions 0.25 and 0.75 step 0.25 with {
				\fill[radius=#1] (0,0) circle;
			}
		}
	},
	dotted pattern/.default={1pt},
}

\newcommand{\gettikzxy}[3]{%
  \tikz@scan@one@point\pgfutil@firstofone#1\relax
  \edef#2{\the\pgf@x}%
  \edef#3{\the\pgf@y}%
}

\newcommand{\lINFlb}{
\pgfdeclarelayer{nodelayer}
	\pgfdeclarelayer{edgelayer}
	
	\pgfsetlayers{edgelayer,nodelayer}
	\tikzstyle{none}=[draw,circle,fill=white,minimum size=.3cm]
	\tikzstyle{empty}=[]
	\tikzstyle{peeled}=[draw,circle,fill=gray!30,minimum size=.5cm]
	\tikzstyle{pruned}=[draw,circle,fill=blue!30,minimum size=.5cm]
 \begin{scope}[xscale=.8,yscale=.8]

	\begin{pgfonlayer}{nodelayer}
		\node [style=none] (0) at (-2, 3) {$a_1$};
		\node [style=none] (1) at (0, 3) {$a_2$};
		\node [style=none] (3) at (1, 1) {$a_3$};
		\node [style=none] (4) at (-1, -0.25) {$a_{i}$};
		\node [style=none] (2) at (-3, 1) {\small $a_m$};
		\node [style=none, label=$b^1_1$] (5) at (-4, 5) {};
		\node [style=none, label=$b^2_1$] (6) at (-3.3, 5) {};
		\node [style=none, label=$b^k_1$] (7) at (-2, 5) {};
		\node [style=none, label=$b^1_2$] (8) at (0, 5) {};
		\node [style=none, label=$b^2_2$] (9) at (.7, 5) {};
		\node [style=none, label=$b^k_2$] (10) at (2, 5) {};
		\node [] (11) at (-5.25, -0.25) {}; 
		\node [] (12) at (-4.5, -0.25) {}; 
		\node [] (13) at (-3.5, -0.25) {}; 
		\node [style=none, label=below:$b^1_i$] (14) at (-2, -1.75) {};
		\node [style=none, label=below:$b^2_i$] (15) at (-1.3, -1.75) {};  
		\node [style=none, label=below:$b^k_i$] (16) at (0, -1.75) {};
		\node [] (17) at (1.75, -0.5) {}; 
		\node [] (18) at (2.5, -0.5) {};
		\node [] (19) at (3, -0.5) {};
	\end{pgfonlayer}
	\begin{pgfonlayer}{edgelayer}
	\path[dotted pattern] (6) -- (7);
	\path[dotted pattern] (9) -- (10);
	\path[dotted pattern] (6) -- (7);
	\path[dotted pattern] (15) -- (16);
		\draw (0.center) -- coordinate[midway] (05) (5.center);
		\node [left,yshift=3pt, xshift=-.2cm] at (05) {$w_0$};

		\draw (0.center) -- coordinate[midway] (06) (6.center);
		\node [right,yshift=3pt, xshift=0cm] at (06) {$w_0$};
		\draw (0.center) -- coordinate[midway] (07) (7.center);
		\node [right,yshift=3.2pt, xshift=0cm] at (07) {$w_0$};
		
		\draw (0.center) to (6.center);
		\draw (0.center) to (7.center);
		\draw[] (1.center) to (10.center);
		\draw[] (1.center) to (8.center);
		\draw[] (1.center) to (9.center);
		\draw[dashed] (3.center) to (17.center);
		\draw[dashed] (3.center) to (18.center);
		\draw[dashed] (3.center) to (19.center);
		\draw[dashed] (2.center) to (13.center);
		\draw[dashed] (2.center) to (12.center);
		\draw[dashed] (2.center) to (11.center);
		\draw (4.center) to (15.center);
		\draw (4.center) to (14.center);
		\draw (4.center) to (16.center);
		\draw[dotted] (2.center) to (4.center);
		\draw[dotted] (4.center) to (3.center);
		\draw (3.center) to (1.center);
		\draw (1.center) to (0.center);
		\draw (0.center) to (2.center);
		\draw (0.center) to (4.center);
		\draw (2.center) to (1.center);
		\draw (2.center) to (3.center);
		\draw (4.center) to (1.center);
		\draw (0.center) to (3.center);
		\draw (1.center) -- coordinate[midway] (13) (3.center);
		\node [right,yshift=0pt, xshift=0cm] at (13) {$w_{ij}$};
	\end{pgfonlayer}
\end{scope}
}

\newcommand{\AvgCase}{

\pgfdeclarelayer{nodelayer}
	\pgfdeclarelayer{edgelayer}
	
	\pgfsetlayers{edgelayer,nodelayer}
	\tikzstyle{none}=[draw,circle,fill=white,minimum size=.3cm]
	\tikzstyle{empty}=[]

\begin{tikzpicture}
	\begin{pgfonlayer}{nodelayer}
		\node [style=none] (0) at (-1.5, 2) {$v_1$};
		\node [style=none] (1) at (0, 3) {$v_2$};
		\node [style=none] (2) at (1.5, 2) {$v_3$};
		\node [style=none] (3) at (1.5, 0) {$v_4$};
		\node [style=none] (4) at (0, -1) {$v_i$};
		\node [style=none] (5) at (-1.5, 0) {$v_n$};
		\node [style=none] (6) at (-3, 1) {$u_i$};
		\node [style=none] (13) at (-3, 3.75) {$u_1$};
		\node [style=none, label=right:$u_{\log n}$] (20) at (-3, -2) {};
		\begin{scope}[xshift=-.7cm]
		\node [style=none, label=right:$b_{i,1}$] (7) at (-4, 2) {};
		\node [style=none, label=right:$b_{i,j}$] (8) at (-4, 1) {};
		\node [style=none, label=right:$b_{i,\log n}$] (9) at (-4, 0) {};
		\node [style=none, label=left:$a_{i,1}$] (10) at (-6, 2) {};
		\node [style=none, label=left:$a_{i,j}$] (11) at (-6, 1) {};
		\node [style=none, label=left:$a_{i,\log n}$] (12) at (-6, 0) {};
		\node [style=none] (g1) at (-5,0) {}; 
		\node [style=none] (g2) at (-5,-1) {};
		\path[dotted pattern] (g1) -- (g2) ;
		\node [style=none] (g3) at (-5,2) {}; 
		\node [style=none] (g4) at (-5,1) {};
		\path[dotted pattern] (g3) -- (g4) ;
		\node [style=none, label=right:$b_{1,1}$] (14) at (-4, 4.75) {};
		\node [style=none, label=right:$b_{1,j}$] (15) at (-4, 3.75) {};
		\node [style=none, label=right:$b_{1,\log n}$] (16) at (-4, 2.75) {};
		\node [style=none, label=left:$a_{1,1}$] (17) at (-6, 4.75) {};
		\node [style=none, label=left:$a_{1,j}$] (18) at (-6, 3.75) {};
		\node [style=none, label=left:$a_{1,\log n}$] (19) at (-6, 2.75) {};
		\node [style=none] (21) at (-4, -1) {};
		\node [style=none] (22) at (-4, -2) {};
		\node [style=none] (23) at (-4, -3) {};
		\node [style=none] (24) at (-6, -1) {};
		\node [style=none] (25) at (-6, -2) {};
		\node [style=none] (26) at (-6, -3) {};

		\end{scope}
		\node [style=none] (27) at (0, 4) {};
		\node [style=none] (28) at (1, 5) {};
		\node [style=none] (29) at (0, 5) {};
		\node [style=none] (30) at (-1, 5) {};
		\node [style=none] (31) at (1, 7) {};
		\node [style=none] (32) at (0, 7) {};
		\node [style=none] (33) at (-1, 7) {};
		\node [style=none] (34) at (3, 1) {};
		\node [style=none] (35) at (4, 0) {};
		\node [style=none] (36) at (4, 1) {};
		\node [style=none] (37) at (4, 2) {};
		\node [style=none] (38) at (6, 0) {};
		\node [style=none] (39) at (6, 1) {};
		\node [style=none] (40) at (6, 2) {};
		\node [style=none] (41) at (0, -2.75) {};
		\node [style=none] (42) at (-1, -3.75) {};
		\node [style=none] (43) at (0, -3.75) {};
		\node [style=none] (44) at (1, -3.75) {};
		\node [style=none] (45) at (-1, -5.75) {};
		\node [style=none] (46) at (0, -5.75) {};
		\node [style=none] (47) at (1, -5.75) {};
	\end{pgfonlayer}
	\begin{pgfonlayer}{edgelayer}
	\draw (19.center) to (15.center);
		\draw (0.center) to (1.center);
		\draw (1.center) to (2.center);
		\draw (2.center) to (3.center);
		\draw [dashed] (3.center) to (4.center);
		\draw [dashed] (4.center) to (5.center);
		\draw (5.center) to (0.center);
		\draw (0.center) to (2.center);
		\draw (0.center) to (3.center);
		\draw (0.center) to (4.center);
		\draw (1.center) to (3.center);
		\draw (1.center) to (4.center);
		\draw (1.center) to (5.center);
		\draw (2.center) to (4.center);
		\draw (2.center) to (5.center);
		\draw (0.center) to (6.center);
		\draw (6.center) to (7.center);
		\draw (6.center) to (8.center);
		\draw (6.center) to (9.center);
		\draw (7.center) to (10.center);
		\draw (7.center) to (11.center);
		\draw (7.center) to (12.center);
		\draw (8.center) to (10.center);
		\draw (8.center) to (11.center);
		\draw (9.center) to (10.center);
		\draw (9.center) to (11.center);
		\draw (9.center) to (12.center);
		\draw (5.center) to (6.center);
		\draw (3.center) to (5.center);
		\draw (13.center) to (14.center);
		\draw (13.center) to (15.center);
		\draw (13.center) to (16.center);
		\draw (14.center) to (17.center);
		\draw (14.center) to (18.center);
		\draw (14.center) to (19.center);
		\draw (15.center) to (17.center);
		\draw (15.center) to (18.center);
		\draw (16.center) to (17.center);
		\draw (16.center) to (18.center);
		\draw (16.center) to (19.center);
		\draw (20.center) to (21.center);
		\draw (20.center) to (22.center);
		\draw (20.center) to (23.center);
		\draw (21.center) to (24.center);
		\draw (21.center) to (25.center);
		\draw (21.center) to (26.center);
		\draw (22.center) to (24.center);
		\draw (22.center) to (25.center);
		\draw (23.center) to (24.center);
		
		\draw [in=330, out=150] (23.center) to (25.center);
		\draw [in=360, out=180] (23.center) to (26.center);
		\draw (27.center) to (28.center);
		\draw (27.center) to (29.center);
		\draw (27.center) to (30.center);
		\draw (28.center) to (31.center);
		\draw (28.center) to (32.center);
		\draw (28.center) to (33.center);
		\draw (29.center) to (31.center);
		\draw (29.center) to (32.center);
		\draw (30.center) to (31.center);
		\draw [in=-120, out=60] (30.center) to (32.center);
		\draw [in=-90, out=90] (30.center) to (33.center);
		\draw (34.center) to (35.center);
		\draw (34.center) to (36.center);
		\draw (34.center) to (37.center);
		\draw (35.center) to (38.center);
		\draw (35.center) to (39.center);
		\draw (35.center) to (40.center);
		\draw (36.center) to (38.center);
		\draw (36.center) to (39.center);
		\draw (37.center) to (38.center);
		\draw [in=150, out=-30] (37.center) to (39.center);
		\draw [in=180, out=0] (37.center) to (40.center);
		\draw (41.center) to (42.center);   
		\draw (41.center) to (43.center);
		\draw (41.center) to (44.center);
		\draw (42.center) to (45.center);
		\draw (42.center) to (46.center);
		\draw (42.center) to (47.center);
		\draw (43.center) to (45.center);
		\draw (43.center) to (46.center);
		\draw (44.center) to (45.center);
		\draw [in=60, out=-120] (44.center) to (46.center);
		\draw [in=90, out=-90] (44.center) to (47.center);
		\draw (8.center) to (12.center);
		\draw (26.center) to (22.center);
		\draw (36.center) to (40.center);
		\draw (43.center) to (47.center);
		\draw (33.center) to (29.center);
	\end{pgfonlayer}
\end{tikzpicture}

}


\newcommand{\beaconLB}{

\pgfdeclarelayer{edgelayer}
	
	\pgfsetlayers{edgelayer,nodelayer}
	\tikzstyle{none}=[draw,circle,fill=white,minimum size=.3cm]
	\tikzstyle{empty}=[]
	\tikzstyle{margin}=[draw,red, thick]

\begin{scope}[xscale=0.8]

	\begin{pgfonlayer}{nodelayer}
		\node [style=none, label=left:$a_1$] (l0) at (-3, 5) {};
		\node [style=none, label=left:$a_2$] (l1) at (-3, 4) {};
		\node [style=none, label=left:$a_3$] (l2) at (-3, 3) {};
		\node [style=none, label=left:$a_n$] (l3) at (-3, 1.5) {};
         
		\node [style=none] (r0) at (0, 5) {};
		
		\node [style=none] (r1) at (0, 4) {};
		\node [style=none] (r2) at (0, 3) {};
		\node [style=none] (r3) at (0, 1.5) {};
		\node[] (lr0) at (1.4,5) {$b_1$};
		\node[] (lr0) at (1.4,4) {$b_2$};
		\node[] (lr0) at (1.4,3) {$b_3$};
		\node[] (lr0) at (1.4,1.5) {$b_n$};

	\end{pgfonlayer}
	\begin{pgfonlayer}{edgelayer}
\foreach \i in {0,1,2,3}{
        \foreach \j in {0,1,2,3}{
        \ifthenelse{\i=\j}{}{
          \draw (l\i.center) to (r\j.center) ;
          }
        }
        
    }
	
		\draw [bend left=50,looseness=1.3] (r0.center) to (r3.center);
		\draw [bend left=50,looseness=1.5] (r0.center) to (r2.center);
		\draw [bend left=50,looseness=1.5] (r1.center) to (r3.center);
		\draw (r0.center) to (r2.center);
		\draw (r1.center) to (r2.center);
		\draw [bend left=50,looseness=1.5] (r2.center) to (r3.center);
		\path[dotted pattern] (r2.south) -- (r3.north);
		\path[dotted pattern] (l2.south) -- (l3.north);
	\end{pgfonlayer}
	 
	\end{scope}
}


\newcommand{\newGrid}[1]{
     \tikzstyle{every node}=[font=\small]


	\tikzstyle{none}=[]
	\tikzstyle{empty}=[]
	\tikzstyle{margin}=[draw,red, thick]
    
    \foreach \i in {0,1,2,3,4,5,6,7} { 
        \node[] (t\i) at (\i,0) {};
        \node[] (m\i) at (\i,-4) {};
        \node[] (b\i) at (\i,-7) {};
        \node[] (l\i) at (0,-\i) {};
        \node[] (mv\i) at (4,-\i) {};
        \node[] (r\i) at (7,-\i) {};
    }
    \node[] (bl) at (7,-7) {};
    
    \foreach \i in {0,1,2,3,4} { 
        \draw[black] (t\i.center) to (\i,-2);
        \draw[black,dotted] (\i,-2) to (\i,-3);
        \draw[black] (\i,-3) to (m\i.center);
        \draw[red,very thick] (m\i.center) to (b\i.center);
        
        \draw[black] (l\i.center) to (2,-\i);
        \draw[black, dotted] (2,-\i) to (3,-\i);
        \draw[black] (3,-\i)  to (mv\i.center);
        \draw[red,very thick] (mv\i.center) to (r\i.center);
    }
    \foreach \i in {5,6,7} { 
        \draw[black] (t\i.center) to (m\i.center);
        \draw[black] (l\i.center) to (mv\i.center);
    }
    \draw[red, very thick] (t0.center) to (m0.center);
    \draw[red, very thick] (t0.center) to (t4.center);
    \draw[red, very thick] (t7.center) to (m7.center);
    \draw[red, very thick] (b0.center) to (b4.center);
    
    \node[] (labelij) at (1.5,-1.5) {$w_{i,j}$};
    \node[] (labeli) at (1.5,.5) {$i$};
    \node[] (labelj) at (-.5,-1.5) {$j$};
    
    \draw [decorate,decoration={brace,amplitude=10pt,mirror,raise=4pt},yshift=0pt] (m7) -- (t7) node [black,midway,xshift=0.8cm] {$m$};
	\draw [decorate,decoration={brace,amplitude=10pt,mirror,raise=4pt},yshift=0pt] (m0) -- (b0) node [black,midway,xshift=-0.8cm] {$k$};
    \draw [decorate,decoration={brace,amplitude=10pt,raise=4pt},yshift=0pt] (t4) -- (t7) node [black,midway,yshift=0.8cm] {$k$};
    \draw [decorate,decoration={brace,amplitude=10pt,raise=4pt},yshift=-0.8cm] (b4) -- (b0) node [black,midway,yshift=-0.8cm] {$m$};
    
    \draw[-latex,black, #1] (3.5,-7) to (3.5,0);
    \draw[-latex,black, #1] (7,-3.5) to (0,-3.5);
}

\newcommand{\GridTuple}{
\begin{scope}[xscale=.6, yscale=.6]
\newGrid{thick, black!30}
\coordinate[draw, circle, blue, fill=blue, inner sep=2] (blue) at (4.5,-3.5);
\coordinate[draw, circle, blue, fill=blue, inner sep=2] (blue) at (6.5,-2.5);
\coordinate[draw, circle, blue, fill=blue, inner sep=2] (blue) at (5.5,-1.5);
\coordinate[draw, circle, blue, fill=blue, inner sep=2] (blue) at (5.5,-0.5);

\coordinate[draw, circle, blue, fill=blue, inner sep=2] (blue) at (1.5,-4.5);
\coordinate[draw, circle, blue, fill=blue, inner sep=2] (blue) at (3.5,-5.5);
\coordinate[draw, circle, blue, fill=blue, inner sep=2] (blue) at (2.5,-6.5);
\coordinate[draw, circle, blue, fill=blue, inner sep=2] (blue) at (0.5,-6.5);
\end{scope}
}



\newcommand{\sparseLB}{
\begin{scope}[yscale=.8]
    \node[draw, circle,label={[xshift=0cm, yshift=0.2cm]$a_{1,0}$}] (mid) at (2,-3.75) {};
    \node[draw, circle, label={[yshift=.cm]$a_1$}] (mid2) at (6,-3.75) {};
    \node[draw, circle, blue, label={[yshift=.1cm,xshift=-.2cm,blue]$a_2$}] (a2) at (6,-8) {};
    \foreach \i in {1.5,3,4.5,6} {
        \node[draw, circle] (l) at (0,-\i) {};
        \node[draw, circle] (r) at (4,-\i) {};
        \draw[black] (l) to (mid);
        \draw[black] (r) to (mid);
        \draw[black] (r) to (mid2);
    }
    \node[] (labela) at (-1.1,-1) {$a_{1,1}$};
    \node[] (labelan) at (-1.1,-6) {$a_{1,4}$};
    \node[] (labela) at (2.8,-1) {$b_{1,1}$};
    \node[] (labelan) at (2.8,-6) {$b_{1,4}$};
    \node[] (labela) at (15.5,-1) {$\bar a_{1,1}$};
    \node[] (labelan) at (15.5,-6) {$\bar a_{1,4}$};
    \draw[blue, dashed, thin] (a2) to (4,-7);
    \draw[blue, dashed, thin] (a2) to (4,-7.66);
    \draw[blue, dashed, thin] (a2) to (4,-8.33);
    \draw[blue, dashed, thin] (a2) to (4,-9);
    
    \begin{scope}[xshift=12.4cm]
        \node[draw, circle,label={[yshift=0.2cm]$\bar{a}_1$}] (mid3) at (0,-3.75) {}; 
        \node[draw, circle,blue, label={[yshift=0.1cm,xshift=.2cm, blue]$\bar{a}_2$}] (bara2) at (0,-8) {}; 
        \foreach \i in {1.5,3,4.5,6} {
            \node[draw, circle] (l1) at (2,-\i) {};
            \draw[black] (l1) to (mid3);
        }        
        \draw[blue, dashed,thin] (2,-7) to (bara2);
        \draw[blue, dashed,thin] (2,-7.66) to (bara2);
        \draw[blue, dashed,thin] (2,-8.33) to (bara2);
        \draw[blue, dashed,thin] (2,-9) to (bara2);
    \end{scope}
\end{scope}

\begin{scope}[yshift=1cm]
     
    \begin{scope}[xshift=9cm,yshift=-5.75cm] 
     \node[circle,black,draw,label] (center) at (0,0) {\small $c_0$};
        \foreach \i in {1,...,4}{
            \ifthenelse{\i=7}
                {}
                {\node[circle,draw, black]  (n\i) at ({\i*360/4}:1.5cm) {};
                \draw (center)--(n\i);
                }
        }

        \node also [label={[xshift=.2cm, yshift=.1cm]$c_4$}] (n2) {};
        \node also [label={above:$c_1$}] (n1) {};
        \node also [label={below:$c_3$}] (n3);
        \node also [label={[xshift=-.2cm, yshift=.1cm]$c_2$}] (n4);

    \end{scope}
        \draw   (mid2) to  [bend right] (n3);
        \draw   (mid2) to  [bend left] (n2);
        \draw   (mid3) to  [bend right] (n1);
        \draw   (mid3) to  [bend right] (n4);
        \draw[blue, dashed, thin] (n1) to [bend left=40] (bara2) ;
        \draw[blue, dashed, thin] (n3) to [bend left=15] (bara2) ;
        \draw[blue, dashed, thin] (n2) to [bend right=20] (a2) ;
        \draw[blue, dashed, thin] (n4) to [bend left=15] (a2) ;

\end{scope}

}



\newcommand{\drawTriangle}{

\node[draw,circle,label=right:{$a_1$}] (a1) at (0,0) {}; 
\node[draw,circle,label=right:{$a_2$}] (a2) at (1.75,-2) {};
\node[draw,circle,label=left:{$a_3$}] (a3) at (-1.75,-2) {};
\draw[black] (a1) to (a2) to (a3) to
(a1);

\foreach \i in {1,2,3} {
  \ifthenelse{\i=3}{\def\col{blue}}{\def\col{white}};
    \node[draw, circle,fill=\col] (b1\i) at (-2+\i,1.5) {};
    \node[draw, circle,fill=\col] (b2\i) at (-.25+\i,-3.5) {};
    \node[draw, circle,fill=\col] (b3\i) at (-3.75+\i,-3.5) {};
    \draw[black] (b1\i) to (a1) ;
    \draw[black] (b2\i) to (a2) ;
    \draw[black] (b3\i) to (a3) ;
} 
}


\newcommand{\drawParen}{
\begin{scope}[yshift=-.6cm]

\node[draw, circle, label=above:$a$] (root) at (0,0) {};
\node[label=above:{\small \textcolor{blue}{ \,$(1,10)$}}] at (0.7,-0.05) {};
\node[draw, circle, label=below:$b$] (cl) at (-1,-1) {};
\node[label=below:{\small \textcolor{blue}{ \,$(2,7)$}}] at (-1,-1.4) {};
\node[draw, circle, label=below:$e$] (cr) at (1,-1) {};
\node[label=below:{\small \textcolor{blue}{ \,$(8,9)$}}] at (1,-1.4) {};
\node[draw, circle, label=below:$c$] (gl) at (-2,-2) {};
\node[label=below:{\small \textcolor{blue}{ \,$(3,4)$}}] at (-2,-2.4) {};
\node[draw, circle, label=below:$d$] (gr) at (0,-2) {};
\node[label=below:{\small \textcolor{blue}{ \,$(5,6)$}}] at (0,-2.4) {};

\draw[black] (root) to (cl); 
\draw[black] (root) to (cr);
\draw[black] (cl) to (gl);
\draw[black] (cl) to (gr);
\end{scope}

\node[] (nums) at (4.9,-0.3) { \small 1\;\;2\;\;3\;4\;\;\;5\;6\;\;7\;\;\;8\;9\;10 };

\node[] (paren) at (5, -1) { \Large(\;(\;(\;)\;(\;)\;)\;(\;)\;)\; };

\draw [blue, decorate,decoration={brace,amplitude=3pt,mirror,raise=4pt},yshift=0pt] (3.87,-1.3) -- (4.37,-1.3) node [black,midway,yshift=-0.4cm] {\blue $c$};
\draw [blue,decorate,decoration={brace,amplitude=3pt,mirror,raise=4pt},yshift=0pt] (4.67,-1.3) -- (5.17,-1.3) node [black,midway,yshift=-0.4cm] {\blue $d$};
\draw [blue,decorate,decoration={brace,amplitude=3pt,mirror,raise=4pt},yshift=0pt] 
(5.91,-1.3) -- (6.41,-1.3) node [black,midway,yshift=-0.4cm] {\blue $e$};
\draw [blue,decorate,decoration={brace,amplitude=8pt,mirror,raise=4pt},yshift=0pt] 
(3.4,-2) -- (5.6,-2) node [black,midway,yshift=-0.7cm] {\blue $b$};
\draw [blue,decorate,decoration={brace,amplitude=8pt,mirror,raise=4pt},yshift=0pt] 
(3,-3) -- (7,-3) node [black,midway,yshift=-0.7cm] {\blue $a$};

}

\newcommand{\InfoLB}{
\pgfdeclarelayer{nodelayer}
	\pgfdeclarelayer{edgelayer}
	
	\pgfsetlayers{edgelayer,nodelayer}
	\tikzstyle{none}=[draw,circle,fill=white,minimum size=.3cm]
	\tikzstyle{empty}=[]
	\tikzstyle{peeled}=[draw,circle,fill=gray!30,minimum size=.5cm]
	\tikzstyle{pruned}=[draw,circle,fill=blue!30,minimum size=.5cm]
 \begin{scope}[xscale=.7,yscale=.7]

	\begin{pgfonlayer}{nodelayer}
		\node [style=none] (0) at (-2, 3) {$a_1$};
		\node [style=none] (1) at (0, 3) {$a_2$};
		\node [style=none] (3) at (1, 1) {$a_3$};
		\node [style=none] (4) at (-1, -0.25) {$a_{i}$};
		\node [style=none] (2) at (-3, 1) {\small $a_m$};
		\node [style=none, label=$b^1_1$] (5) at (-4, 5) {};
		\node [style=none, label=$b^2_1$] (6) at (-3.3, 5) {};
		\node [style=none,fill=blue, label=$b^k_1$] (7) at (-2, 5) {};
		\node [style=none, label=$b^1_2$] (8) at (0, 5) {};
		\node [style=none,fill=blue, label=$b_2^2$] (9) at (.7, 5) {};
		\node [style=none, label=$b^k_2$] (10) at (2, 5) {};
		\node [] (11) at (-5.25, -0.25) {}; 
		\node [] (12) at (-4.5, -0.25) {}; 
		\node [] (13) at (-3.5, -0.25) {}; 
		\node [style=none,fill=blue, label=below:$b^1_i$] (14) at (-2, -1.75) {};
		\node [style=none, label=below:$b^2_i$] (15) at (-1.3, -1.75) {};  
		\node [style=none, label=below:$b^k_i$] (16) at (0, -1.75) {};
		\node [] (17) at (1.75, -0.5) {}; 
		\node [] (18) at (2.5, -0.5) {};
		\node [] (19) at (3, -0.5) {};
	\end{pgfonlayer}
	\begin{pgfonlayer}{edgelayer}
	\path[dotted pattern] (6) -- (7);
	\path[dotted pattern] (9) -- (10);
	\path[dotted pattern] (6) -- (7);
	\path[dotted pattern] (15) -- (16);
		\draw (0.center) -- coordinate[midway] (05) (5.center);
		\node [left,yshift=3pt, xshift=-.2cm] at (05) {1};

		\draw (0.center) -- coordinate[midway] (06) (6.center);
		\node [right,yshift=3pt, xshift=0cm] at (06) {1};
		\draw (0.center) -- coordinate[midway] (07) (7.center);
		\node [right,yshift=3.2pt, xshift=0cm] at (07) {1};
		\draw[] (1.center) to (10.center);
		\draw[] (1.center) to (8.center);
		\draw[] (1.center) to (9.center);
		\draw[dashed] (3.center) to (17.center);
		\draw[dashed] (3.center) to (18.center);
		\draw[dashed] (3.center) to (19.center);
		\draw[dashed] (2.center) to (13.center);
		\draw[dashed] (2.center) to (12.center);
		\draw[dashed] (2.center) to (11.center);
		\draw (4.center) to (15.center);
		\draw (4.center) to (14.center);
		\draw (4.center) to (16.center);
		\draw[dotted] (2.center) to (4.center);
		\draw[dotted] (4.center) to (3.center);
		\draw (3.center) to (1.center);
		\draw (1.center) to (0.center);
		\draw (0.center) to (2.center);
		\draw (0.center) to (4.center);
		\draw (2.center) to (1.center);
		\draw (1.center) -- coordinate[midway] (13) (3.center);
		\node [right,yshift=0pt, xshift=0cm] at (13) {$w_{ij}=6\cdot(2^b+\delta_{ij})-2$};
		\draw (2.center) to (3.center);
		\draw (4.center) to (1.center);
		\draw (0.center) to (3.center);
	\end{pgfonlayer}

\end{scope}
}

\begin{abstract}
A* is a classic and popular method for graphs search and path finding. It assumes the existence of a heuristic function $h(u,t)$ that estimates the shortest distance from any input node $u$ to the destination $t$. Traditionally, heuristics have been handcrafted by domain experts.  However, over the last few years, there has been a growing interest in learning heuristic functions. Such learned heuristics estimate the distance between given nodes based on ``features'' of those nodes.

In this paper we formalize and initiate the study of such feature-based heuristics. In particular, we consider heuristics induced by norm embeddings and distance labeling schemes, and provide lower bounds for the tradeoffs between the number of dimensions or bits used to represent each graph node, and the running time of the A* algorithm. We also show that, under natural assumptions, our lower bounds are almost optimal.
\end{abstract}

\section{Introduction}
A* is a classic and popular method for graphs search and path finding. It provides a method for computing the shortest path in a given weighted graph $G=(V,E,W)$ from a source $s$ to a destination $t$ that is often significantly faster than classic algorithms for this problem. It assumes the existence of a heuristic function $h(u,t)$ that estimates the shortest distance from any input node $u$ to the destination $t$. The algorithm uses the greedy approach, at each step selecting a node $u$ that minimizes $d(s,u)+h(u,t)$, where $d(s,u)$ is the (already computed) distance from $s$ to $u$. Alternatively, A* can be viewed as a variant of Dijkstra algorithm, with the distance function $d(u,v)$ replaced by $d(u, v) + h(v,t) - h(u,t)$. Since its inception in the 1960s, the algorithm has found many application, e.g., to robotics~\citep{yonetani2020path,bhardwaj2017learning}, game solving~\citep{cui2011based}, computational organic chemistry~\citep{chen2020retro}. Over the last two decades, it has been also shown be highly effective for ``standard'' shortest path computation tasks in road networks~\citep{goldberg2005computing}. 

The performance of A* is governed by the quality of the heuristic function that estimates the distance from a given node to the target. For example, if the heuristic function is perfect, i.e., $h(u,t)=d(u,t)$ for all nodes $u$, then the number of vertices scanned by A* is proportional to the number of hops in the shortest path. In practice, the heuristic function is carefully selected based on the properties of the underlying class of problems. E.g., for graphs whose vertices corresponds to points in the plane (e.g., shortest paths in road networks or paths avoiding 2D obstacles), typical choices include Euclidean, Manhattan or Chebyshev distances. In many cases, the task of identifying an appropriate heuristic for a given problem can be quite difficult. 

Over the last few years, there has been a growing interest in {\em learning} heuristic functions based on the properties of input graphs, see e.g.,~\citep{yonetani2020path, bhardwaj2017learning, chen2020retro}. Such predictors are trained on a collection of input graphs. After training, the predictor estimates the distance between given nodes based on the pre-computed ``features'' of those nodes, plus possibly other auxiliary information (see Section~\ref{s:info} for a discussion). The features are vectors in a $d$-dimensional space, either handcrafted based on the domain knowledge, or trained using machine learning methods. Thus, such learned heuristics bear similarity to  {\em metric embeddings}~\citep{naor2018metric} or {\em distance labels}~\citep{gavoille2004distance}, two notions that have been extensively investigated in theoretical computer science and mathematics. The notions of the quality of embeddings in the above lines of research is, however,  quite different from what is required in the A* context. Specifically, a typical objective of metric embeddings or distance labelling is to preserve the distances between every pair of vertices up to some (multiplicative) approximation factor. On the other hand, in the context of A* search,  it is acceptable if the estimated distances  deviate significantly from the ground truth,   as long as the A* search process that uses those estimates does not take too much time. At the same time, approximate estimates of the distances in the metric embedding sense might not be useful if they are not able to disambiguate between many different short paths, leading to high running times. This necessitates  studying different notions of embeddings and labeling that are tailored to the A* context. 

Motivated by these considerations, in this paper we formalize and initiate the study of  embeddings and distance labeling schemes that induce efficient A* heuristics. In particular, we study tradeoffs between the dimensionality of the embeddings (or the length of the labels in labeling schemes) and the complexity of the A* search process. Our focus is on {\em average-case} performance of A*, where the average is taken over all pairs of vertices in the graph. For lower bounds,  average-case results provide much stronger limitations than worst-case results. At the same time, those lower bounds naturally complement our algorithms, which rely on the average-case assumptions for technical reasons.

We start by formally defining the metric that we will use to evaluate the quality of a heuristic. 

\begin{definition}
For a graph $G=(V,E,W)$, and a heuristic $h$,  $P(s,t)$ denote the set of vertices on the shortest path between $s$ and $t$ with the maximal number of hops, and  let $S_h(s,t)$ be the set of vertices scanned by A* given $s,t$ and $h$. Further let $p(s,t)=|P(s,t)|$.
We say that $h$ has {\em an additive overhead} $c$ {\em on average} if
\[ \mb{E}_{s,t \in V}[|S_h(s,t)|-p(s,t)]|\le c \]
where $s$ and $t$ are chosen independently and uniformly at random from $V$.
\end{definition}

Informally, the above definition measures the number of ``extra'' vertices that needs to be scanned, in addition to the size of the shortest path. We note that one could alternatively define the overhead in a ``multiplicative'' way, by computing the ratio between the number of vertices scanned and the number of hops in a shortest paths. However, several of our results (esp. the lower bounds) rely on instances  where we are guaranteed that all shortest paths have constant number of hops. In this case, an additive overhead of $T$ automatically translates into a multiplicative overhead of $\Omega(T)$.

In this paper we focus on heuristics that are  {\em consistent}. The consistency property, introduced by  ~\citet{4082128} and stated below, implies that the A* search algorithm using the heuristic correctly identifies the shortest path upon termination.

\begin{definition}[Consistent heuristic function]\label{def:consistent_embedding}
A heuristic function $h(s,t)$ which estimates $dist(s,t)$ is  \emph{consistent} if for any destination $t$ and edge $(u,v)$, the modified cost $w_t(u,v)=w(u,v)+h(v,t)-h(u,t)$ is non-negative and also $h(t,t)=0$.
\end{definition}

\subsection{Heuristic types}
We now formally define the types of heuristics that we study in this paper. Specifically, we introduce heuristic induced by norm embeddings and labelling schemes, as well as embeddings induced by ``beacons'', which can be viewed as a special type of norm embeddings induced by the $\ell_{\infty}$ norm.

\begin{definition}[Norm heuristic]\label{def:norm_embedding}
A  heuristic function $h: V \times V \to \mathbb{R}$ is a {\em norm heuristic} induced by a function  $\pi:V\to \mathbb{R}^d$ and the $l_p$ norm  if $h(s,t)=\|\pi(s)-\pi(t)\|_p$ for any pair  of vertices $s,t \in V$. The dimension $d$ is referred to as the {\em dimension of $h$}. 
\end{definition}

Note that $h(s,t)$ as defined above satisfies triangle inequality, i.e. $h(s,u)+h(u,t)\ge h(s,t)$ for any $s,t,u\in V$.

\begin{definition}[Labeling heuristic]\label{def:labelling_embedding}
A  heuristic function $h: V \times V \to \{0 \ldots C\}$ is a {\em labeling heuristic} induced by functions $f:V\to \{0 \ldots C\}^L$ and $g:\{0 \ldots C\}^{L}\times \{0 \ldots C\}^{L}\to \{0 \ldots C\}$, if $h(s,t)=g(f(s),f(t))$ for any pair of vertices $s,t \in V$. Here the parameter $C$ is assumed to be polynomial in $|V|$, and $L$ is referred to as the {\em label length}.
\end{definition}

 Note that $f$ and $g$ must be fixed functions working for any graph and label pairs. Furthermore,  we assume no  limitation on the computational power  of $f$ and $g$.
Finally, note that the function $g$ does not have to be a norm function, or to satisfy the triangle inequality.

 We also define a particular type of heuristic that is induced by a set of ``beacons''. This is a natural class of A* heuristics that is particularly popular in the context of shortest path problems~\citep{abraham2010highway}.

\begin{definition}[Beacon heuristic]\label{def:beacon_embedding}
 A heuristic function $h: V\times V\to \mathbb{R}$ is a {\em beacon heuristic} induced by a set $B \subset V$ if   $h(s,t)=\|\pi(s)-\pi(t)\|_{\infty}$, where $\pi(s)=(dist(s,b_1),\ldots,dist(s,b_{|B|}))$ (so that  $h(s,t)=\max_{b\in B}| dist(s,b)-dist(t,b)|$).
\end{definition}

Given the above definitions, we can now formally state our results (also summarized in Table~\ref{t:results}).

\renewcommand\arraystretch{1.5}
\colorlet{lbColor}{gray!10}

\begin{table}[b!]
\begin{center}
\begin{NiceTabular}{ |c|c|c|c|c| }
\hline
 Type  & Settings &  Comments & Result  &  Reference \\ \hline
  \multirow{6}{*}{\shortstack[c] {Norm \\ heuristics}}
& $p<\infty$ & & $d=o(n/\log n) \Rightarrow \tilde \Omega(n)$ avg. overhead 
   & \Block{2-1}{Thm.~\ref{thm:lpnorm_informal}}\\  %
& $p=\infty$ &  & $d=o(\log n) \Rightarrow \tilde \Omega(n)$ avg. overhead    & \\   
\cline{2-5}
&  & general graphs & \Block{2-1}{$d=o(n^{\alpha}) \Rightarrow \tilde \Omega(n^{1-\alpha})$ avg. overhead } & Thm.~\ref{thm:linftylb_informal}\\ \cline{3-3}\cline{5-5}
& $p=\infty$ & grid graphs ($\alpha<0.5$)& & Thm.~\ref{thm:linftylb_grid} \\
\cline{3-5}
  &  weighted graphs & \cellcolor{lbColor} beacon based & \cellcolor{lbColor} \Block{2-1}{ $d=O(n^{\alpha})$ and $O(n^{1-\alpha})$ avg. overhead} &   \cellcolor{lbColor} \Block{2-1}{Thm.~\ref{thm:upper_bound_informal}}\\
  & &  \cellcolor{lbColor}Assumption~\ref{asm:break-tie} &  \cellcolor{lbColor}& \cellcolor{lbColor} \\
\hline
\multirow{3}{*}{\shortstack[c] {Labeling \\ heuristics}} & \Block{3-1}{weighted graphs} & general graphs  & \Block{2-1}{$L=o(n^{\alpha})\Rightarrow  \Omega(n^{1-\alpha})$ avg. overhead} & Thm.~\ref{thm:information_lb_informal} \\
\cline {3-3} \cline {5-5}
& & grid graphs ($\alpha<0.5$) & & Thm.~\ref{thm:information_lb} \\
\cline{3-5}
& & \cellcolor{lbColor} Assumption~\ref{asm:stronger_unique_shortest_path} & \cellcolor{lbColor} $L=O(n^\alpha)$ and $O(n^{1-\alpha})$ avg. overhead & \cellcolor{lbColor} Thm.~\ref{thm:upper_bound_not_norm_informal} \\
\hline
\end{NiceTabular}
\caption{Summary of our results. Upper bounds are colored in gray. }
\label{t:results}
\end{center}
\end{table}

\subsection{Lower bounds results} Our first set of results provide lower bounds on the trade-off between the complexity of a heuristic and the complexity of the A* algorithm. Specifically, we will show that there exist graphs where all shortest paths have constant complexity (i.e., consist of a constant number of hops), but on which A* will scan a large number of vertices unless its heuristic has large complexity. 
We start from heuristics induced by norms.
Recall that  the average-case complexity is defined with respect to the set of all possible source-destination pairs.

\begin{restatable}[Lower bound for norm heuristics]{theorem}{normlbgeneral}\label{thm:lpnorm_informal}
There exists an $n$-vertex unweighted graph $G=(V,E)$ of constant diameter such that for any consistent norm heuristic induced by an $l_p$ norm of dimension $d$, where $d=o(n/\log n)$ for $p<\infty$ or $d=o(\log n)$ for $p=\infty$,
the A* algorithm scans $\tilde{\Omega}(n)$ vertices on average.
\end{restatable}

The proof for the above theorem appears in Section~\ref{sec:general_lb}. For the special case of the $l_{\infty}$ norm we show a stronger lower bound. However, it requires that the input graph is {\em weighted} (see Section~\ref{sec:linfty_lb}).

\begin{restatable}[Lower bound for $l_{\infty}$ norm heuristics]{theorem}{linftylbGeneral}\label{thm:linftylb_informal}
There exists a weighted $n$-vertex graph $G$ such that (i) the shortest path between any pair of vertices in $G$ consists of a constant number of hops and (ii) for any consistent norm heuristic induced by $l_{\infty}$ of dimension $d=o(n^{\alpha})$, the A* algorithm scans  $\Omega(n^{1-\alpha})$ vertices on average. 
\end{restatable}

We now turn to heuristics induced by general labeling schemes. Our main result here is that a trade-off analogous to that in Theorem~\ref{thm:linftylb_informal} applies even to general labeling schemes \tnew{(see Section~\ref{sec:information_lb}).}

\begin{restatable}[Lower bound for labeling heuristics]{theorem}{thmInfoGeneral}
\label{thm:information_lb_informal}
For any consistent labeling heuristic function $h$ with length $L=o(n^{\alpha})$, there exists an $n$-vertex weighted graph $G$ such that (i) the shortest path between any pair of vertices in $G$ consists of a constant number of hops and (ii) the A* algorithm scans  $\Omega(n^{1-\alpha})$ vertices on average.
\end{restatable}

For the last two theorems, we also show their variants for a graph $G$ that is  a 
grid graph (with weighted edges). See Theorems~\ref{thm:linftylb_grid} and~\ref{thm:information_lb} in Appendix~\ref{sec:grid} for further details.
This makes our instances realizable in the robotics planning problems, which is one of the main applications where heuristics have been applied~\cite{bhardwaj2017learning} (see Section \ref{s:info} for more details).

\subsection{Upper bound results}

We complement the lower bounds from the previous sections with upper bounds, i.e., constructions of embedding or labeling heuristics. Our constructions, however, require certain assumptions regarding breaking ties in the A* algorithm. Specifically, at any step of the algorithm, there could be multiple vertices with the same value of the estimate $d(s,u)+h(u,t)$, and breaking such ties appropriately  can have significant impact on the efficiency of A* \citep{asai2016tiebreaking,10.1613/jair.5249,10.5220/0008119405950604}. In our theoretical analysis we deal with tie-breaking in two ways. Our first approach is to make an assumption about the {\em A* algorithm}. Specifically,  Assumption~\ref{asm:break-tie} states that the A* algorithm can break ties in a way that minimizes its complexity. Our second approach is to make an assumption about the {\em input graph}. Specifically, Assumption~\ref{asm:stronger_unique_shortest_path} states that the weighted graph $G$ is such that there is a unique shortest path between any pair of vertices, and that there is a separation between the lengths of the shortest and the second shortest path. We note that the lower bounds corresponding to our upper bounds are consistent with these assumptions, as per discussions in the corresponding sections.

We start by stating the assumption about the tie-breaking behavior of the algorithm.

\begin{assumption}[Breaking ties]\label{asm:break-tie}
We assume that, when there is a tie on the estimated distance lower bound (as specified in Definition~\ref{def:est_lb}) for multiple different vertices, the A* algorithm chooses the next vertex in a way that minimizes the total running time. 
\end{assumption}

Note that in the context of Fact~\ref{lm:scan_condition} stated in Preliminaries, the assumption states that A* only scans those vertices that {\em must} be scanned.

Under this assumption, we show that Theorem~\ref{thm:linftylb_informal} is almost optimal. In fact, the upper bound is based on a beacon heuristic (see Section~\ref{sec:beacon_ub}). 

\begin{restatable}[Beacon-based upper bound]{theorem}{upperBoundBeacon}\label{thm:upper_bound_informal}
Under Assumption~\ref{asm:break-tie}, for any weighted graph $G=(V,E,W)$ with $|V|=n$ vertices, there exists a beacon heuristic induced by a set $B$ with $|B|=n^{\alpha}$ with an additive overhead $n^{1-\alpha}$, i.e.,  such that  A* scans at most $n^{1-\alpha}+\mb{E}_{s,t}\left[|P(s,t)|\right]$ vertices on average.
\end{restatable}

We now consider the assumption about lack of ties in the input graph. 

\begin{assumption}[Unique shortest path]\label{asm:stronger_unique_shortest_path}
We assume that for any pair $(s,t)$ of vertices in $G$, the difference between the length of their shortest path and second shortest path between $s$ and $t$ is greater than a constant $c>0$ (e.g. $c=3$).
\end{assumption}

Under this assumption, we show that Theorem~\ref{thm:information_lb_informal} is almost optimal \tnew{(see Section~\ref{sec:notnorm_ub}).}

\begin{restatable}
[Labeling-based upper bound]{theorem}{upperBoundLabeling}\label{thm:upper_bound_not_norm_informal}
Under Assumption~\ref{asm:stronger_unique_shortest_path}, for any weighted graph $G=(V,E,W)$ with vertex set size $|V|=n$, there exists a consistent labeling heuristic of length $L = O(n^{\alpha})$ with an additive overhead $n^{1-\alpha}$, i.e.,  such that A*  scans at most $n^{1-\alpha}+\mb{E}_{s,t}[|P(s,t)|]$ vertices on average.
\end{restatable}

\subsection{Learned A* overview}
\label{s:info}
As discussed previously, the main motivation for this study are recent works on learned A* heuristic functions, where the learned features on which the $h$ heuristic is based can be viewed as embedding or labeling schemes.
In this section we briefly describe several approaches used in recent papers on A* with learned heuristic functions~\citep{yonetani2020path, bhardwaj2017learning,chen2020retro}. 
As these are applied machine learning papers, our description here abstracts away many  technical details, focusing on high-level ideas. 
\begin{itemize}
 \item \citep{bhardwaj2017learning} is focused on learning A* heuristics for robot  path planning. The experiments are done on graphs induced by two-dimensional environments modeled as 2D grids with obstacles. The heuristic function is represented by  a feed-forward neural network with two fully connected hidden layers containing respectively 100 and 50 units with ReLu activation.  The input to the network is a 17 dimensional feature vector. The features contain the descriptor of the node $u$ (containing various features of the spatial point) as well as the descriptor of the current ``state'' of the neural A* algorithm when reaching $u$. 
 
    \item \citep{yonetani2020path} is also focused on robot path planning in 2D grids with obstacles.  It uses a neural network architecture called  U-net to transform the input instance into a ``guidance map'', which for each graph vertex provides an adjustment to a ``standard'' heuristic function based on  Chebyshev distance between the given node and the destination node $t$. The input to U-net consists of the input graph, the source node $s$ and destination node $t$. The presence of the last component technically makes it possible to map each node into a one-dimensional scalar equal to the true distance to target node $t$. However, given that training of the heuristic function is expensive, it is natural to aim for a heuristic that is ``target independent'', i.e., trained for a particular environment, independently of the choice of the source and destination.  In this way one can reuse the same heuristic trained for a particular grid for many source/destination pairs.
    
    \item \citep{chen2020retro}  addresses problems in a different domain, focusing on path searching in the molecule space. The goal of the planner is to identify a series of re-actions that can lead to the synthesis of a target product, a classic problem in organic chemistry. The feature vector used as an input to the learned heuristic is based on the ``Morgan fingerprint of radius 2'', which is 2048 bit long. The vector is then fed into a single-layer fully connected neural network of hidden dimension 128, which provides an estimate of the distance. 
\end{itemize}

In the context of the aforementioned works, our results provide insights into the complexity of {\em representing} efficient heuristics for various classes of graphs. We leave other important issues, such as the complexity of {\em learning} those heuristics, to future work.

\section{Preliminaries}

\paragraph{Notation.} We use $dist(u,v)$ to denote the length of the shortest paths between nodes $u$ and $v$ in a graph. For $(u,v)\in E$ we use $w(u,v)$ to denote the edge length.

\paragraph{Heuristic properties.} We first recall several standard definitions of heuristic functions which will be useful in our proofs.

\begin{definition}[Admissible heuristic function]\label{def:admissible_embedding}
For a heuristic function $h(s,t)$ which estimates $dist(s,t)$, it is said to be \emph{admissible} if for any source and destination $(s,t)$, we have $h(s,t)\le dist(s,t)$.
\end{definition}

It is  well-known that all consistent heuristic functions (Definition~\ref{def:consistent_embedding}) are admissible.

\begin{definition}[Sub-additive heuristic function]\label{def:subadditive_embedding}
For a heuristic function $h(s,t)$ which estimates $dist(s,t)$, it is said to be \emph{sub-additive} if for any three vertices $(u,v,w)$, we have $h(u,v)+h(v,w)\ge h(u,w)$.
\end{definition}

Note that consistent heuristics may not be sub-additive, but all norm heuristics (Definition~\ref{def:norm_embedding}) and beacon heuristics (Definition~\ref{def:beacon_embedding}) are sub-additive.

\begin{definition}[Estimated distance lower bound]\label{def:est_lb}
Given a query pair $dist(s,t)$,
for any vertex $u$,
we define $g(u)$ to be the \emph{estimated distance lower bound} calculated by A* for shortest path from $s$ to $t$ passing $u$:  $g(u)=dist(s,u)+h(u,t)$. In particular, $g(t)=dist(s,t)$.
\end{definition}

\paragraph{A* algorithm.} Now, we describe the operation of  A*  with heuristic function $h$, which will be the main object of our analysis in the paper. (See Algorithm~\ref{alg:astar_algorithm})

\begin{algorithm2e}[!ht]
    \caption{\algoname}\label{alg:astar_algorithm}
    \LinesNumbered
    \SetAlgoNoLine
    \DontPrintSemicolon
    \KwIn{$G(V,E,W)$, start point: $s$, end point: $t$} 
    Initialize $S=\{s\}$, $d(s)=0$, $d(S\setminus s)=\infty$\;
    \While{$t\notin S$}{
        $u\gets\argmin\limits_{u\in V\setminus S}\{d(u)+h(u,t)\}$\;
        \For{$v$ s.t. $(u,v)\in E$}{
            $d(v)=\min(d(v),d(u)+w(u,v))$\;
        }
        $S=S\cup \{u\}$\;
    }
\end{algorithm2e}

\begin{fact}[Scan condition, see \citet{4082128}]\label{lm:scan_condition}
Suppose we apply A* to calculate $dist(s,t)$ with a consistent heuristic function $h(u,t)$. Then a vertex $u$ {\em must} be scanned if and only if $dist(s,u)+h(u,t)<dist(s,t)$ and {\em may} be scanned if and only if $dist(s,u)+h(u,t)\le dist(s,t)$.
\end{fact}

\section{Lower bounds}
\subsection{An average-case lower bound for $l_p$ norms}\label{sec:general_lb}

In this section, we prove Theorem~\ref{thm:lpnorm_informal}.

\normlbgeneral*

Before exactly defining the lower bound instance (depicted in Figure~\ref{fig:sparseLB}) and proving the theorem, we give a short overview of the ideas behind the construction and proof.

\paragraph{Proof overview.} 
The central part of our constructed hard instance contains a star with $n$ petals. 
We then add to it $O(\log n)$ additional sets of vertices, each of size $O(n)$, in a way that ensures that for every pair of star petals $(u,v)$, there exists $\Omega(n^2)$ vertex pairs in the additional sets whose shortest path goes through $(u,v)$.
Given this construction, we rely on a   standard ``packing argument"  stating that a low dimensional space cannot contain a large equilateral vertex set. 
Therefore, it holds that  if the dimension $d$ of the embedding is too small, then
by the admissibility 
property of the embedding, for at least one pair of the star petals, its embedded distance is distorted.  
We then utilize this pair of distorted petals to show that  every query $(s,t)$ whose shortest path goes through the distorted pair has a large query overhead.
Specifically, we rely on the design of non-uniqueness of the shortest paths to further amplify the penalty due to the distorted pair to $\Omega(n)$ per each query as above, so that in total we get that there exists $\Omega(n^2)$ pairs for which the query overhead is $\Omega(n)$.

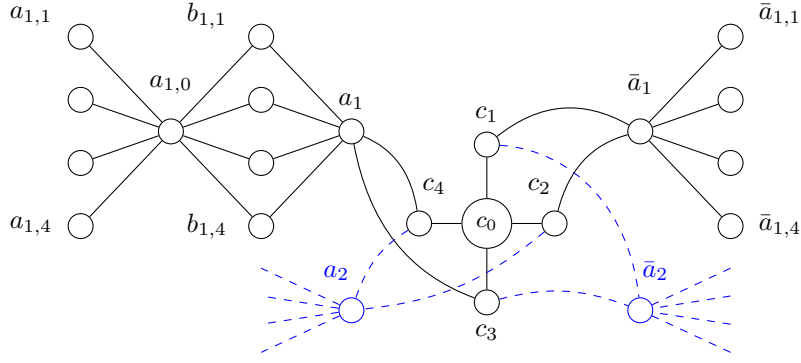
\begin{figure}[ht!]
\centering
\begin{tikzpicture}

\begin{scope}[xscale=0.6,yscale=.7]
\sparseLB    
\end{scope}

\end{tikzpicture}
\caption{Average-case complexity lower bound instance for $l_p$ norms with $k=2$ and $n=4$.}
\label{fig:sparseLB}
\end{figure}

\paragraph{The lower bound instance.} We continue to formally define our lower bound instance for heuristics induced by $l_p$ norms. Note that the lower bound instance is both sparse and unweighted.
 The lower bound graph contains a star with  $n$ leaves $\{c_1,\ldots,c_n\}$ connected to a center $c_0$. Without loss of generality, we assume $n=2^k$ is a power of $2$. Then we have $k$ pairs of vertices $\{(a_i,\bar{a}_i)\}_{i=1}^k$. For each pair $a_i$ and $\bar{a}_i$, and for each $c_j$, we connect $c_j$ to $a_i$ if the $i$-th bit in the binary representation of $j-1$  is $1$ and to $\bar{a}_i$ otherwise. 
Then, for each $a_i$, we create auxiliary vertex sets $\{a_{i,j}\}_{j=0}^n$ and $\{b_{i,j}\}_{j=1}^n$. We connect 
$a_{i,0}$ to all $\{a_{i,j}\}_{j=1}^n$, and, for each $j$, we connect $b_{i,j}$ to $a_{i,0}$ and $a_i$.  Similarly, for each $\bar{a}_i$, we connect it to all $\{\bar{a}_{i,j}\}_{j=1}^n$.
Our instance has size  $|V|=\Theta(n\log n)$ and $|E| = O(|V|)$. See Figure~\ref{fig:sparseLB} for an illustration with $k=2$ and $n=4$.

\begin{lemma}\label{lem:no_simplex}
Suppose $h$ is a norm heuristic function $h(s,t)=\|\pi(s)-\pi(t)\|$. For the instance described above with $|V|=\Theta(n\log n)$,  as long as one pair of distinct vertices in $\{c_i\}_{i=1}^n$, say $c_i$ and $c_j$,  satisfies  $h(c_i,c_j)=\|\pi(c_i)-\pi(c_j)\|<2$ ,  the A* algorithm scans at least $\Omega\left(\frac{n}{\log^2 n}\right)$ vertices on average.
\end{lemma}

\begin{proof}
Suppose $\|\pi(c_u)-\pi(c_v)\|<2$ and $u\neq v$, then there exists a digit $p$ s.t. their binary representations differ, e.g. $u$'s digit is $0$ and $v$'s is $1$. Then, for the pair $(a_p,\bar{a}_p)$ whose real distance is $dist(a_p,\bar{a}_p)=4$, by sub-additivity, their embedded distance satisfies
$$\|\pi(a_p)-\pi(\bar{a}_p)\|\le \|\pi(a_p)-\pi(c_u)\|+\|\pi(c_u)-\pi(c_v)\|+\|\pi(c_v)-\pi(\bar{a}_p)\|<4=dist(a_p,\bar{a}_p).$$
Now, we can check that for any query pair $(s,t)$ with $s\in \{a_{p,j_1}\}_{j_1=1}^n$ and $t\in \{\bar{a}_{p,j_2}\}_{j_2=1}^n$, A* must scan the whole set $\{b_{p,j_1}\}_{j_1=1}^n$  %
To see this, first note that all vertices $u \in \{b_{p,j_1}\}_{j_1=1}^n$ 
lie on the shortest path between $s$ and $t$, and satisfy $h(u,t)<dist(u,t)$ (due to the sub-additivity and since $\|\pi(a_p)-\pi(\bar{a}_p)\|<dist(a_p,\bar{a}_p)$). Therefore, every $u\in \{b_{p,j_1}\}_{j_1=1}^n$  will be scanned by A* according to Fact~\ref{lm:scan_condition}. To calculate the average query complexity, 
observe that for a random query pair $(s,t)$, with probability at least $\Omega(\frac{1}{\log^2 n})$ we have that $s\in\{a_{p,j_1}\}_{j_1=1}^n$ and $t\in\{\bar{a}_{p,j_1}\}_{j_1=1}^n$ simultaneously hold. Once this event occurs, A* will scan at least $\Omega(n)$ vertices, so the average query complexity is lower bounded by $\Omega(\frac{n}{\log^2 n})$.
\end{proof}

\begin{definition}\label{def:equilateral_set}
In a metric space $l^d_p$, we say a vertex set $X$ is equilateral if any two different vertices $x,y\in X$ satisfy $\|x-y\|_p=1$. We define $e(l^d_p)$ to be the size of the largest equilateral vertex set in $l^d_p$ and its inverse function $e^{-1}(p,n)$ equals to the minimum $d$ such that $e(l^d_p)\ge n$.
\end{definition}

\begin{lemma}
For a metric space $l^d_p$, we have the following upper bounds on $e(l^d_p)$
$$e(l^d_p)\le\left\{
\begin{array}{lll}
O(d\log d) & p=1 & \textup{\citep{Alon_equilateralsets}}\\
d+1 & p=2 & \textup{\citep{10.2307/2975549}}\\
O(d\log d) & 2<p<\infty & \textup{\citep{Swanepoel04aproblem}}\\
2^d & p=\infty & \textup{\citep{10.2307/2975549}}
\end{array}\right.$$
and corresponding lower bounds for $e^{-1}(p,n)$
$$e^{-1}(p,n)\ge\left\{
\begin{array}{ll}
\Omega(\frac{n}{\log n}) & p=1\\
n-1 & p=2\\
\Omega(\frac{n}{\log n}) & 2<p<\infty\\
\log_2 n & p=\infty
\end{array}\right.$$
Here our asymptotic bound is in terms of $d$ and hides the dependence on $p$.
\end{lemma}

\begin{proof}[Proof of Theorem~\ref{thm:lpnorm_informal}]
We only need to verify that the assumption made in Lemma~\ref{lem:no_simplex} is true. By the definition of an equilateral set (in Definition~\ref{def:equilateral_set}) and a consistent heuristic function (in Definition~\ref{def:consistent_embedding}), we know that for metric space $l^d_p$ if $d<e^{-1}(p,n)$, then for a set of $n$ 2-equidistant vertices $\{c_i\}_{i=1}^n$ in the original graph, at least two of them will have embedded distance smaller than $2$.
\end{proof}

\subsection{A stronger lower bound  for \texorpdfstring{$l_{\infty}$}{linf} norm space}\label{sec:linfty_lb}

In this section, we prove  a stronger lower bound for embeddings into $l_{\infty}^d$ both for general graphs and for grid graphs. The first lower bound instance is a general graph with constant diameter (as in the previous section), 
which means that the ratio of the number of vertices scanned by A* to the actual path length is large. Based on this instance we prove Theorem~\ref{thm:linftylb_informal}, restated below. 

\linftylbGeneral*
The second lower bound instance has super-constant diameter, but it is a grid graph, which is a more natural structure in the context of robot planning applications. This instance is described in Appendix~\ref{sec:linftylb_grid_proof}, and is used to  prove the grid variant of Theorem~\ref{thm:linftylb_informal}.

\begin{restatable}{theorem}{linftylbGrid} \label{thm:linftylb_grid}
There exists a weighted grid graph $G$ so that 
any consistent norm heuristic $h(s,t)$ induced by $\pi:V\to \mathbb{R}^d$ with dimension  $d<o(n^{\alpha})$ for $\alpha<0.5$ and the $l_{\infty}$ norm will result in A* average-case query complexity of at least $\Omega(n^{1-\alpha})$.
\end{restatable}

Importantly,  Theorem~\ref{thm:linftylb_informal}, as well as its grid analog,  Theorem~\ref{thm:linftylb_grid}, hold even under Assumption~\ref{asm:break-tie}, i.e.,  when A* can break ties arbitrarily. This is because Lemma~\ref{lm:alwayscan} below holds for {\em any} tie-breaking rule, i.e., it lower bounds the number of vertices that {\em must} be scanned.
Therefore,  the upper bound Theorem~\ref{thm:upper_bound_informal}  is a matching upper bound. For further details on the upper bound see  Section~\ref{sec:beacon_ub}.

We start with the intuition behind the lower bound construction and proof.

\paragraph{Proof idea.}
The proof idea here follows the one in the previous section, where we first prove that there exist some pairs of vertices with large distortion, and then amplify the query cost penalty due to these pairs.
 The previous equilateral set argument works well for $l_p$ norm with $p<\infty$ but works poorly for the $l_{\infty}$ norm. Therefore, we use a special property of $l_{\infty}$ to get a tighter lower bound, specifically, that   the $l_{\infty}$ norm solely depends on one coordinate (with the maximum absolute value). We can observe that if the distance that our embedding ``wants to memorize" is random, then it cannot remember too much information in each one dimensional space. 
Our construction consists of a clique of size  $O(m)$, 
and the edge weights are  chosen uniformly in $[10,11]$. This ensures that the shortest path between any pair of vertices is the edge directly connecting them, and that the weights are randomized and independent of each other. Therefore, each dimension can only ``memorize'' $O(m)$ shortest path distances, while
there are $O(m^2)$ pairs of vertices.
Similarly to the previous case, the second 
step of amplifying the penalty is achieved 
by adding  auxiliary sets to the main lower bound construction, so that each distorted 
pair appears in the shortest path for many queries $(s,t)$.

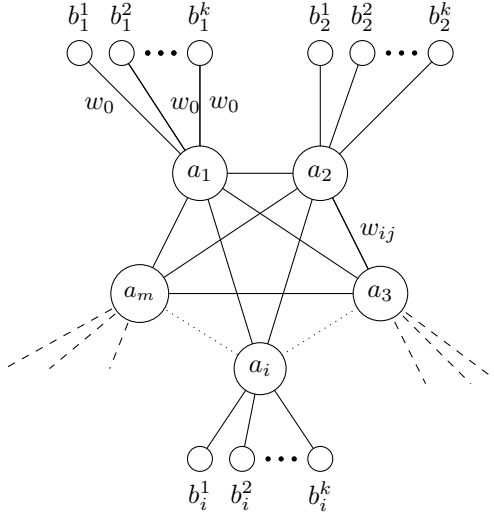
\begin{figure}[ht!]
\centering
\begin{tikzpicture}
\lINFlb
\end{tikzpicture}
\caption{Average-case complexity lower bound instance for $l_{\infty}$ norms.}
\label{fig:instance_linf_nongrid}
\end{figure}

\paragraph{The lower bound instance.} Our lower bound instance (in Figure~\ref{fig:instance_linf_nongrid}) consists of a clique $\{a_i\}_{i=1}^m$ where each $a_i$ is attached to $k$ leaves $\{b_{i,j}\}_{j=1}^k$. The edge weight between $(a_i,a_j)$ is $w_{ij} = 10+u_{ij}$, where $u_{ij}$ are i.i.d. random variables chosen from the uniform distribution over the interval $[0,1]$. The edge weight between each attached leaf to its corresponding vertex in the clique is equal to $w_0$, specified subsequently.

\begin{definition}[Approximated-tie]\label{def:approx_tie}
For a set of values $\{w_i\}_{i=1}^{n}$ and an parameter $\epsilon$, we say there exists an approximated-tie if there exists a coefficient vector $\{c_i\}_{i=1}^{n}$ where each $c_i\in \{-4 \ldots 4\}$ with at least one non-zero component and $|\sum_{i=1}^{n}c_iw_i|\le\epsilon$.
\end{definition}

\begin{proposition}\label{lm:existence_epsilon_nongrid}
For weights $\{w_{i,j}\}_{i,j\in [m]}$ generated as above, there exists an $\epsilon>0$ such that no approximated-tie exists in $\{w_{i,j}\}$ with probability at least 0.99.
\end{proposition}

Let $\epsilon$ be as in the above proposition We set the edge weights between the clique to the leaves to be $w_0=\frac{\epsilon}{16n}$. From now on, we assume that there is no $\epsilon$ approximated-tie in the graph.
We can observe that in this graph, the shortest path between any two vertices in $\{a_i\}$ is exactly the edge connecting them.

\begin{lemma}
\label{lm:alwayscan}
For a pair of vertices $(a_i,a_j)$, if the embedding distance between them has error larger than $\frac{\epsilon}{2n}$, then any query $(b^p_i,b^q_j)$ for $p,q\in [k]$ will always scan $\Omega(k)$ vertices (under any tie-breaking rule.)
\end{lemma}

\begin{proof}
In the setting of Theorem~\ref{thm:linftylb_informal}, we denote $h(s,t)=\|\pi(s)-\pi(t)\|_{\infty}$. %
If $\|\pi(a_i)-\pi(a_j)\|\le w_{i,j}-\frac{\epsilon}{2n}$, we can check that for query $dist(b^p_i,b^q_j)$ all vertices in $\{b^{p'}_i\}_{p'\in[k]}$ will be scanned by A*:
Let $b^{p'}_i$ be an arbitrary leaf attached to $a_i$.
\begin{align*}
&\quad dist(b^p_i,b^{p'}_i)+\|\pi(b^{p'}_i)-\pi(b^q_j)\|\\
&\le dist(b^p_i,b^{p'}_i)+\|\pi(b^{p'}_i)-\pi(a_i)\|+\|\pi(a_i)-\pi(a_j)\|+\|\pi(a_j)-\pi(b^q_j)\|\\
&\le 2w_0+w_0+w_{i,j}-\frac{\epsilon}{2n}+w_0\\
&\le w_{i,j}-\frac{\epsilon}{4n}\\
&< dist(b^p_i,b^q_j)
\end{align*}
Therefore, for every $p'\in [k]$, by the scan condition, $b_i^{p'}$ will be scanned. 
\end{proof}

\begin{lemma}\label{lm:half}
For the heuristic considered by Theorem~\ref{thm:linftylb_informal}, as long as $d\le o(m)$, at least half of pairs $(s,t)$ where $s,t \in \{a_i\}, i=1 \ldots m$, satisfy $h(s,t)<dist(s,t)-\frac{\epsilon}{2n}$.
\end{lemma}

\begin{proof}
By the definition of $l_{\infty}$ norm, if $\|\pi(u)-\pi(v)\|_{\infty}\ge dist(u,v)-\frac{\epsilon}{2n}$, then there must exist a coordinate $i\in[d]$ where $|\pi(u)_i-\pi(v)_i|\ge dist(u,v)-\frac{\epsilon}{2n}$. We call such coordinate a ``crucial coordinate'' for the pair of vertices $(u,v)$

We claim that for each coordinate $i\in [d]$, at most $m-1$ pairs of vertices use coordinate $i$ as their crucial coordinate. To show this,  we first construct an auxiliary graph with vertices $\{a_i\}$ where we assign an edge to $(u,v)$ if $(u,v)$ uses the current coordinate as their ``crucial coordinate''. Because the auxiliary graph has $m$ vertices, if there are more than $m$ edges, at least one simple cycle exists with length $l\le m$. We can go along the cycle to get an approximated-tie consisting of a $\pm 1$ weighted sum of $w_{i,j}$. Because $l\cdot \frac{\epsilon}{2n}\le \epsilon$, such a cycle would violate Proposition~\ref{lm:existence_epsilon_nongrid}.

Therefore if there are at most $o(m)$ coordinates, there can be at most $o(m^2)$ ``crucial coordinates'', which means that at least $\frac{m^2}{4}$ pairs are embedded with error more than $\frac{\epsilon}{2n}$.
\end{proof}

\begin{proof}[Proof of Theorem \ref{thm:linftylb_informal}]
We set $m=n^{\alpha}$ and $k=n^{1-\alpha}$ then $|V|=\Theta(n)$.
Now we calculate the overall average complexity lower bound. Each pair $(a_i,a_j)$ with $h(a_i,a_j)<dist(a_i,a_j)-\frac{\epsilon}{2n}$  contributes $k^2$ pairs of queries, each of which causes A* to have query complexity $\Omega(k)$. Therefore the average query complexity is lower bounded by the following:
$$\Omega\left(\frac{\frac{m^2}{4}\cdot k^2\cdot k}{n^2} \right) =\Omega\left(\frac{n^{2\alpha}\cdot n^{2(1-\alpha)}\cdot n^{1-\alpha}}{n^2}\right)=\Omega(n^{1-\alpha}).$$
\end{proof}

\subsection{An average-case lower bound for labeling heuristics}
\label{sec:information_lb}

In this section, we state  two lower bounds for labeling heuristics. The first holds for general graphs, and the second for grid graphs. Note 
that  labeling heuristics doe not have the sub-additivity property.

First, we make some assumptions on the graphs we consider. As usual, we consider weighted undirected graphs $G(V,E,W)$ where $|V|=n$. We further assume that the edge weight between any two vertices $u$ and $v$, $w(u,v)$, is an integer and that the sum of weights satisfies $\sum_{u,v} w(u,v)\le C$, where $C=poly(n)$. Therefore,  each quantity produced by our algorithm can be represented by $B=O(\log n)$ bits.

We first recall the  lower bound that holds for general, non grid, graphs.

\thmInfoGeneral*

We also provide a grid graph which preserves the complexity lower bound and is common in robotics planning problem.

\begin{restatable}{theorem}{thmInfoGrid}
\label{thm:information_lb}
For any consistent labeling heuristic function $h$ with length $L=o(n^{\alpha})$, there exists an $n$-vertex weighted grid graph $G$ such that the A* algorithm scans  $\Omega(n^{1-\alpha})$ vertices on average.
\end{restatable}

Importantly, both the graph used in Theorem~\ref{thm:information_lb_informal} and the one used to prove  Theorem~\ref{thm:information_lb} satisfy Assumption~\ref{asm:stronger_unique_shortest_path} regarding unique shortest paths,  after an appropriate weight scaling. Therefore,  Theorem~\ref{thm:upper_bound_not_norm_informal} shown later in Section~\ref{sec:notnorm_ub} is a  matching upper bound result.

In this section we only prove the general theorem, and we defer the proof of Theorem~\ref{thm:information_lb} to Appendix~\ref{sec:info_lb_grid}, as it follows similar ideas.  We start from a short overview.

\paragraph{Proof idea.} This proof inherits the idea of reducing a good embedding to a ``memorizing random number" task. We utilize the hardness of the  ``indexing'' communication complexity problem to prove that it is impossible to remember too many random numbers using a limited amount of storage.

\paragraph{The lower bound instance.} 

The lower bound instance  is depicted in Figure~\ref{fig:instance_info_nongrid}. It consists of a clique $\{a_i\}_{i=1}^m$ and additionally, each $a_i$ has a corresponding leaf set $\{b_i^p\}_{p=1}^k$. To avoid boundary cases, we assume both $m,k\ge 10$. Indeed, this
graph has the same structure as that in  Figure~\ref{fig:instance_linf_nongrid},
however, we choose a different set of parameters. The edge weight between each attached leaf to its corresponding vertex in the clique is $1$ and the edge weight between $(a_i,a_j)$ is $w_{ij}=6\cdot(2^b+\delta_{ij})-2$, where $\delta_{ij}\in[0,2^{b}-1]$ to be specified later. The setting of $w_{ij}$ guarantees that the shortest distance between any two vertices $a_i,a_j$ is due to the  edge $(a_i,a_j)$. Finally, we define an $m$-tuple  to be a set consisting of $m$ vertices $(v_1,\ldots, v_m)=(b^{i_1}_1, \ldots, b^{i_{m}}_m)$ for indices $i_j\in[k]$. Therefore, there are $k^m$ different $m$-tuples. See in Figure~\ref{fig:instance_info_nongrid}  an example to one possible $m$-tuple colored in blue.

\begin{figure}[ht!]
\centering
\begin{tikzpicture}
 \InfoLB
\end{tikzpicture}
\caption{Average-case complexity lower bound instance for labeling.}
\label{fig:instance_info_nongrid}
\end{figure}
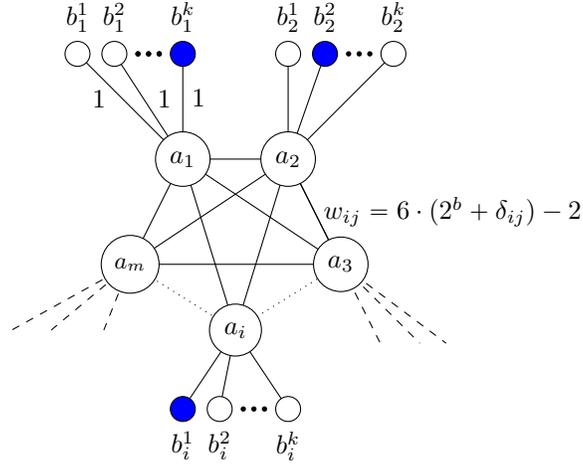

For the ease of proof, we define the binary version of labeling heuristics which only allows for $0/1$ bits and the length is multiplied by $B$.

\begin{definition}[Binary labeling heuristics]\label{def:lembedding_01}
Labeling heuristic with length $L$ consists of two fixed deterministic functions $f:V\to \{0,1\}^{BL}$ and $g:\{0,1\}^{BL}\times \{0,1\}^{BL}\to \{0,1\}^B$. The distance of two vertices $u,v$ in the embedded space is defined to be $g(f(u),f(v))$.  
\end{definition}

Interpreting   strings in $\{0,1\}^B$ as the binary version of integers within the range $[0,2^B-1]$, the consistent definition still applies to the new definition.

We will reduce our embedding problem to the following well known indexing problem.

\begin{definition}[Indexing problem]\label{def:general_index}
Alice gets a vector, $x\in [0,2^{b}-1]^n$ chosen uniformly at random. Bob gets an index $i\in[1,n]$ chosen uniformly at random. The goal is for Bob to report $x_i$ after receiving a single message from Alice. 
\end{definition}

We set the bit length for each number in the input of indexing problem to be $b=\frac{B}{8}$.
Another mild condition is that $B$ should be set  such that $2^b<n$ and $n^6< 2^B$,  to ensure that  every quantity produced in the reduction process can still be represented in $B$ bits.

\begin{theorem}[Folklore, see  e.g.,~\cite{rao2020communication}]
\label{thm:result_general_index}
Any one-way protocol for the indexing problem defined above requires $\Omega(nb)$ bits of communication in order to succeed with  probability $>\frac12$. 

\end{theorem}

\begin{definition}[A bad pair]\label{def:bad_pair} We say that a pair of vertices $v_i,v_j$ is a \emph{bad pair} if they have large distortion, defined as $g(f(v_i),f(v_j))<dist(v_i,v_j)-3$.
\end{definition}

We continue to argue about the contribution of bad pairs to the average query complexity of A*.

\begin{lemma}\label{lem:bad_pairs_contribution}
If there are $\Omega(m^2)$ bad pairs in every $m$-tuple, the contribution of bad pairs to the query complexity of A* (over all possible pair queries) is $\Omega(m^2k^3)$, so the average complexity for A* is at least $\Omega(k)$.
\end{lemma}
\begin{proof}
We first argue that if a pair of vertices $(b^p_i,b^q_j)$ is a bad pair (as defined in Definition~\ref{def:bad_pair}), then for any shortest path query of the form $(b_i^{p'}, b^q_j)$ for any $p'\in[k]$ will scan $b^p_i$.
By Fact~\ref{lm:scan_condition}, for a shortest path query  $(s,t)$, a vertex $u$ is always scanned by A* if $dist(s,u)+g(f(u),f(t))<dist(s,t)$. By Definition~\ref{def:bad_pair},  if $(b^p_i,b^q_j)$ is a bad pair, then $g(f(b^p_i),f(b^q_j))< dist(b^p_i,b^q_j)-3=w_{ij}-1$. Therefore,  for any query from $b^{p'}_i$ to $b^q_j$ where $b^{p'}_i$ is any other element in the set $\{b^p_i\}_{p=1}^k$
\begin{align*}
dist(b^{p'}_i,b^p_i)+g(f(b^p_i),f(b^q_j))< 2+w_{i,j}-1<dist(b^{p'}_i,b^q_j).
\end{align*}
implying that $b_i^{p}$ will be scanned.
Therefore, every bad pair contributes $\Omega(k)$ cost to the summation of all-pair shortest path queries  cost. 

It remains  to lower bound the number of such bad pairs at $\Omega(k^2m^2)$.
By the assumption that there are $\Omega(m^2)$ bad pairs for each $m$-tuple,  counting with repetitions, there are at least $\Omega(m^2 \cdot k^m)$
bad pairs. 
Since every bad pair appears in at most  $k^{m-2}$ $m$-tuples, it follows that  there are at least $\Omega(k^2\cdot m^2)$ distinct bad pairs.
Hence, the bad pairs contribute $\Omega(m^2k^3)$ to the overall query complexity of A*.
Because $n=\Theta(m\times k)$, the average-case complexity is at least $\Omega(k)$.
\end{proof}

\begin{lemma}\label{lm:memorizeweight_nongrid}
Consider the graph described above and the A* algorithm using a consistent $L=o(m)$ labeling heuristic. If for any choice of $\{\delta_{ij}\}$ there exists an $m$-tuple 
$(v_1,...v_m)$ (where recall that $v_i=b_i^{p_i}$ for some $p_i\in[k]$), for which at most $o(m^2)$ number of the induced all-pair embedding distances $\{g(f(v_i),f(v_j))\}_{i,j\in[m]}$ have large distortion, $g(f(v_i),f(v_j))<dist(v_i,v_j)-3$, then there is a one-way protocol for the indexing problem using $o(nb)$ bits and succeeds with probability at least $1-o(1)$.
\end{lemma}

\begin{proof}
Suppose that for every choice of $\{\delta_{ij}\}$, there exists at least one $m$-tuple $(v_1, \ldots, v_m)$ with its induced all-pair embedding distance satisfying that for at least $1-o(1)$ fraction of the pairs $v_i,v_j$, $dist(v_i,v_j)-3\le g(f(v_i),f(v_j))\le dist(v_i,v_j)$. Then we construct the following protocol for the indexing problem. We let $n=\binom{m}{2}$, and we think of the index that bob receives as a tuple $(i,j)$, where $i,j\in [m]$ and $i< j$.

Alice: given an input $x\in [0,2^b-1]^{\binom{m}{2}}$, set $\{\delta_{ij}\}=x$ respectively. Enumerate over all the $m$-tuples to find the one whose embedded distance has distortion smaller than $3$ for at least $1-o(1)$ fraction of all pairs. Note that a distortion smaller than 3  means that $dist(a_i,a_j)-1\le g(f(v_i),f(v_j))\le dist(a_i,a_j)+2$. Recall that  $dist(a_i,a_j)=w_{ij}=6\cdot(2^b+\delta_{ij})-2$. Hence,  $\left\lceil\frac{g(f(v_i),f(v_j))}{6}\right\rceil-2^{b}=\delta_{ij}$. Alice then sends to Bob  the embedding $\{f(v_i)\}_{i=1}^m$,  which has bit length $o(m^2B)$, or equivalently $o(m^2b)$ in terms of the indexing bit length $b$.

Bob: Receive $\{f(v_i)\}$ from Alice, use them to construct $\delta'_{ij}=\left\lceil\frac{g(f(v_i),f(v_j))}{6}\right\rceil-2^{b}$ where at least $1-o(1)$ fraction of them satisfy $\delta'_{ij}=\tnew{\delta_{ij}=x}$. 
Given Bob's query $(i,j)$, Bob 
 answers the corresponding term $\delta'_{ij}$.

Therefore, the protocol above can solve the indexing problem with probability greater than $1-o(1)$ and communication complexity $o(m^2b)$ for any input.
\end{proof}

We are now ready to prove Theorem~\ref{thm:information_lb_informal}.

\begin{proof}[Proof of Theorem~\ref{thm:information_lb_informal}]
We set $m=n^{\alpha}$ and $k=n^{1-\alpha}$, implying that  that the size of the  lower bound graph is $\Theta(m\cdot k)= \Theta(n)$ as required. 
We prove this theorem by contradiction. If there exists a consistent labeling heuristic $h$ with length $L=o(n^{\alpha})$ such that for any choice of $\{\delta_{ij}\}$, A* with heuristic $h$ scans $o(n^{1-\alpha})$ vertices. Then by Lemma~\ref{lem:bad_pairs_contribution}, for each choice of $\{\delta_{ij}\}$ we can find a $m$-tuple, where there is at most $o(1)$ fraction of bad-pairs. Finally, by Lemma~\ref{lm:memorizeweight_nongrid}, we can get an algorithm for the indexing problem, contradicting Theorem~\ref{thm:result_general_index}.
\end{proof}

\section{Upper bounds}

\subsection{Random beacons}\label{sec:beacon_ub}

In this subsection, we analyze the performance of beacon heuristic under Assumption~\ref{asm:break-tie} when A* has the ability to arbitrarily break ties. 
Assumption~\ref{asm:break-tie} means that for the scanning condition stated in Fact~\ref{lm:scan_condition}, when answering query $dist(s,t)$, A* will only scan vertex $u$ that {\em must} be scanned, i.e., if either $u\in P(s,t)$ or $dist(s,u)+h(u,t)<dist(s,t)$.

We can verify that the random beacon embedding is consistent according to Definition~\ref{def:consistent_embedding}.

\upperBoundBeacon*

\begin{proof}
First, we uniformly at random select the set $B$ from the vertex set $V$. By Assumption~\ref{asm:break-tie}, A* answering query $dist(s,t)$ will visit a vertex $u \notin P(s,t)$ if $dist(s,u)+h(u,t)<dist(s,t)$, or equivalently $\|\pi(u)-\pi(t)\|< dist(s,t)-dist(s,u)$. By the definition of $\|\pi(u)-\pi(t)\|=\max_{v\in B}|dist(u,v)-dist(t,v)|$ in the random beacon embedding, $u$ will satisfy the inequality above  only if $s$ can give a better distance estimation for   $dist(u,t)$ than all beacons in $B$. 
This happens with probability at most $\frac{1}{|B|+1}$ over a random selection of $B \cup \{s\}$ for a fixed pair $u,t$. In expectation, for any pair $(u,t)$, there are at most $\frac{n}{|B|}$  such vertices $s$. 
Therefore, taking a summation over all pairs of $(u,t)$ yields that the expected number of extra vertices scanned is at most  $\frac{n^2\cdot n/|B|}{n^2}=\frac{n}{|B|}$. Finally, we choose $|B|=n^{\alpha}$ to get the desired result.
\end{proof}

\subsection{A labeling heuristic}\label{sec:notnorm_ub}

In this section we present a labeling heuristic with parameters as in Theorem~\ref{thm:upper_bound_not_norm_informal}, assuming  Assumption~\ref{asm:stronger_unique_shortest_path}.
First, we observe that this assumption  holds for ``generic'' graphs, i.e., holds with high probability assuming that the edge weights are generated randomly. We defer the proof to Appendix~\ref{appen:prob_unique_proof}.

\begin{proposition}\label{prop:prob_unique}
For any weighted graph $G=(V,E,W)$ with $|V|=n$ with edge weight independently generated from a discrete distribution whose probability mass at each point is upper bounded by $\frac{1}{n^5}$, Assumption~\ref{asm:stronger_unique_shortest_path} holds with probability at least $1-\frac{1}{n}$.
\end{proposition}

Since it is easy to re-scale the edge weights, we set the constant in Assumption~\ref{asm:stronger_unique_shortest_path} to be $3$ for the ease of later proofs.
For convenience, we restate the main theorem shown in this section.

\upperBoundLabeling*

\paragraph{Proof idea.}
Since a ``naive" beacon embedding  produces ties, our idea for the proof is to slightly increase the embedded distance  so that potential ties are avoided by adding more information to the embedding.
 This  preserves the consistency of the heuristic,  except in the case where the original embedded distance is already tight. The latter  case only happens when  one of the three vertices ($s,u$ or $t$) lies on the shortest path between the other two. Under the unique shortest path assumption (Assumption~\ref{asm:stronger_unique_shortest_path}), this case is easy to detect  if we encode the Euler tour of the shortest path tree rooted at each beacon in the ``extra'' bits of the labels. At the same time, adding this information increases the label length by only a negligible factor.

In order to prove the theorem, we first extend the definition of the beacon heuristic from the previous section.

\begin{definition}[Beacon heuristic with tie breaking]\label{def:beacon_embedding_tie_breaker}
For a beacon set $B$, we define $\pi$ to be a concatenation of two $d$ dimensional embeddings $\pi^0$ and $\pi^1$ where $\pi^0_i(s)=dist(s,b_i)$ is the same as beacon-based embedding and $\pi^1_i(s):V\to [2n]^2$  are the locations of the two occurrences of the vertex $s$ in the Euler tour of the shortest path tree rooted at $b_i$ (see Figure~\ref{fig:parseq}). \begin{align}
h(s,t)=\max_{i\in|B|}\left(\left|\pi^0_i(s)-\pi^0_i(t)\right|+\underbrace{\left|\sum_{j\in\{0,1\}} sign\left(\pi^1_{i,j}(s)-\pi^1_{i,j}(t)\right)\right|}_{\textup{checking whether $s$ is $t$'s ancestor/descendent}}\right)\label{line:pi_design}
\end{align}
\end{definition}

Note that by the property of the Euler tour, the second part in Equation~\ref{line:pi_design} is equal to either $0$ or $2$. Moreover, the value is equal to $0$ only when $u$ is an ancestor of $v$ or $v$ is an ancestor of $u$.  

Note $h(\cdot,\cdot)$ is not a norm function of $\pi(\cdot)-\pi(\cdot)$ in general.  

\begin{figure}[ht!]
\centering
\begin{tikzpicture}
\begin{scope}[xscale=0.8,yscale=.8]
\drawParen
\end{scope}
\end{tikzpicture}
\caption{An example of a shortest path tree (left) and its Euler tour (right).}
\label{fig:parseq}
\end{figure}
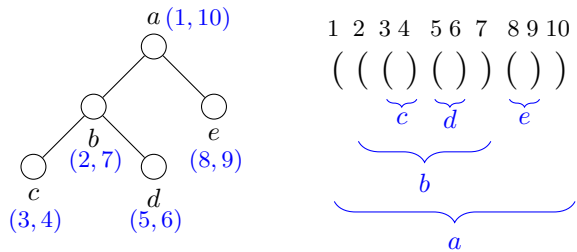

\begin{lemma}
Under Assumption~\ref{asm:stronger_unique_shortest_path}, the beacon-based embedding with tie breaker defined in Definition~\ref{def:beacon_embedding_tie_breaker} is consistent and has the following compact expression:
\begin{align}
h(s,t)=\begin{cases}dist(s,t)\ &if\ \|\pi^0(s)-\pi^0(t)\|_{\infty}=dist(s,t)\\ \|\pi^0(s)-\pi^0(t)\|_{\infty}+2\ &if\ \|\pi^0(s)-\pi^0(t)\|_{\infty}<dist(s,t)\end{cases}\label{line:goal_pi}
\end{align}
\end{lemma}

\begin{proof}
We prove the equation in line~\ref{line:goal_pi} first.  Due to the unique shortest path assumption, we can show that $|dist(b_i,s)-dist(b_i,t)|=dist(s,t)$ if and only if $s$ is an ancestor or descendant of $t$ on the shortest path tree rooted at $b_i$.  The ``if'' part is immediate, since one of the vertices $s$ and $t$ lies on the shortest path between the root and the other vertex. For the ``only if'' part, let $u$ be the lowest common ancestor of $s$ and $t$ on the shortest path tree rooted at $b_i$ and assume w.l.o.g. that $dist(u,s)>dist(u,t)$. The shortest paths from $s$ to $u$ and from $t$ to $u$ consist of the tree edges. Because $|dist(b_i,s)-dist(b_i,t)|=dist(s,t)$, by canceling the common sub-path, we get that $dist(u,s)-dist(u,t)=dist(s,t)$. Therefore, concatenating the shortest paths from $s$ to $t$ and from $t$ to $u$ produces a second shortest path from $s$ to $u$, thus violating our unique shortest path Assumption~\ref{asm:stronger_unique_shortest_path}. Thus, the second expression in line~\ref{line:pi_design} is equal to $0$ when $|dist(b_i,s)-dist(b_i,t)|=dist(s,t)$ and equal to $2$ 
otherwise. Therefore Equation~\ref{line:goal_pi} follows.

To prove that the new heuristic is consistent, we first prove that $h(s,t)\le dist(s,t)$, i.e., $h$ is admissible. By Equation~\ref{line:goal_pi}, it suffices to show that if $|dist(b_i,s)-dist(b_i,t)|<dist(s,t)$ then $|dist(b_i,s)-dist(b_i,t)|\le dist(s,t)-2$.
We prove this by contradiction. If there exists $b_i$ such that $dist(s,t)-2<dist(b_i,s)-dist(b_i,t)<dist(s,t)$, then we can concatenate paths $P(b_i,t)$ and $P(t,s)$ to get a path from $b_i$ to $s$ of length smaller than $dist(b_i,s)+2$, contradicting Assumption~\ref{asm:stronger_unique_shortest_path} that the second shortest path should be longer than the shortest path by at least $3$. 

Now, we fix any two vertices $u,v$ and recall the definition of consistency: $h(u,t)\le w(u,v)+h(v,t)$. By Equation~\ref{line:goal_pi}, we know that $h(u,v)-\|\pi^0(u)-\pi^0(v)\|_{\infty}$ equals to either $0$ or $2$, so there are 4 cases for checking the inequality.
Consider another heuristic $h'(u,v)=\|\pi^0(u)-\pi^0(v)\|_{\infty}$. It is a standard beacon heuristic and thus consistent. Therefore we only need to consider the worst case where $h(u,t)=h'(u,v)+2$ and $h(v,t)=h'(v,t)$. Note that $h(v,t)=h'(v,t)$ means that $h(v,t)=dist(v,t)$, so the RHS becomes $w(u,v)+dist(v,t)$ which is greater than $dist(u,t)$ and therefore  greater than $h(u,t)$ by admissibility shown above.
\end{proof}

\begin{proof}[Proof of Theorem~\ref{thm:upper_bound_not_norm_informal}]
The analysis extends on the proof of Theorem~\ref{thm:upper_bound_informal} by using random beacon-based embedding which is $\pi^0$ in Definition~\ref{def:beacon_embedding_tie_breaker}.

Considering a pair of query $dist(s,t)$, for those vertices $u$ satisfying $dist(s,u)+\|\pi^0(u)-\pi^0(t)\|_{\infty}<dist(s,t)$, they are inevitably scanned, and by invoking  Theorem~\ref{thm:upper_bound_informal} with $O(n^{\alpha})$ dimension, the average number of such vertices $u$ is upper bounded by $O(n^{1-\alpha})$. Also, for those vertices $u$ satisfying $dist(s,u)+\|\pi(u)-\pi(t)\|_{\infty}>dist(s,t)$, they will not bother our A* because $h(s,v)\ge \|\pi^0(s)-\pi^0(t)\|_{\infty}$ is a strictly tighter distance estimation than $\pi^0$. Now we only care about those $u$ satisfying $dist(s,u)+\|\pi(u)-\pi(t)\|_{\infty}=dist(s,t)$ and there are two cases: If $\|\pi(u)-\pi(t)\|_{\infty}=dist(u,t)$, we have $dist(s,u)+dist(u,t)=dist(s,t)$, so $u$ lies on $P(s,t)$ and should be scanned by A*. Otherwise $\|\pi(u)-\pi(t)\|_{\infty}<dist(u,t)$ and $h(u,t)=\|\pi(u)-\pi(t)\|_{\infty}+2$, breaking the tie, so A* will not scan such a vertex $u$. Therefore, all previous ties impose no extra scanning complexity for A*.
\end{proof}

\bibliographystyle{plainnat}
\bibliography{ref}

\newpage
\appendix

\section{Lower bounds for grid graphs}\label{sec:grid} \label{sec:linftylb_grid_proof}

The main technical difficulty for the grid version lower bounds is that we are no longer allowed to construct a clique as in the proofs of the general theorems. Hence, the independence property of each shortest path distance is lost. Fortunately, with a careful choice of parameters, we can restrict the shortest path to follow a certain trajectory, and then we can use similar ideas to the ones for the general case lower bounds.

\subsection{Stronger lower bound for $l_{\infty}$ norm on a grid}
In this section we modify our lower bound instance so that is based on a 2D grid. This makes it more similar to the instances in the path planning applications discussed in the introduction. Instead of drawing the instance in the form of vertices connected by weighted edges, it will be more natural to associate vertices with grid nodes (cells), and put weights on vertices. The edges from each vertex/cell lead to up to four of its neighbors in the grid, and are unweighted. 

\begin{figure}[ht!]
\centering
\begin{tikzpicture}
\begin{scope}[xscale=.6, yscale=.6]
\newGrid{thick, black!30}
\end{scope}
\end{tikzpicture}
\caption{Grid lower bound instance. 
} 
\label{fig:instance_linf_grid}
\end{figure}

Our instance (call it $G$) is depicted in Figure~\ref{fig:instance_linf_grid}. It consists of a concatenation of 3 rectangles. The upper-left rectangle $R$ is of dimension $m\times m$ and the two rectangles attached to its right ($R_r$) and lower ($R_l$) sides are of dimensions $m\times k$ and $k \times m$,  where $m$ and $k$ are defined later. %
For the simplicity of notation, we label the lower right node of $R$ with coordinates $(0,0)$ and all other nodes are labeled with coordinates $(x,y)$, where the x-axis is horizontal and pointing to the {\em left}, while  %
the y-axis is vertical and pointing to the top. We denote the  node with coordinates $(x,y)$ by $v_{x,y}$, and denote its weight by $w_{x,y}$; the values of the weights will be defined later. Finally, the edges corresponding to the red lines in the picture (interpreted as obstacles) are removed. 
We use $n=|V|$ to denote the size of the graph. This instance has size $|V|=n=m^2+2mk$. We set the weights $w_{i,j}$ as follows:
$$w_{x,y}=\left\{
\begin{array}{ll}
(n-x)\cdot n^3+ u_{x,y} & x,y\ge 0\\
\epsilon_w & x \textup{ or } y<0
\end{array}\right.$$
where $u_{x,y}$ are i.i.d. random variables chosen from the uniform
distribution over the interval $[0,1]$ 
and $\epsilon_w$ is a small number to be specified later.

Recall we aim to prove Theorem~\ref{thm:linftylb_grid}.
\linftylbGrid*

\begin{proposition}\label{lm:existence_epsilon_grid}
For weights $\{w_{x,y}\}$ generated from the description above, we can find a constant $\epsilon$ where no approximated-tie exists in $\{w_{x,y}\}$ with probability at least 0.99.
\end{proposition}

Now, we set the parameter $\epsilon_w=\frac{\epsilon}{8mk}$ where $\epsilon$ is defined by Proposition~\ref{lm:existence_epsilon_grid}.

\begin{lemma}\label{lm:auxiliary_graph}
Consider an auxiliary graph with vertices $V'=\{a_i\}_{i=0}^{m-1}\cup\{b_i\}_{i=0}^{m-1}$ and edges $(a_i,b_j)$ with weights equal to  the distance $dist(v_{i,0},v_{0,j})$ in the original graph. Then, for any simple cycle $C$ in this auxiliary graph, any $\pm 1$ weighted sum 
of the weights of the nodes in the cycle has absolute value larger than $\epsilon$.
\end{lemma}

\begin{proof}
The proof is similar to the proof of Lemma~\ref{lm:half}.
First, from the choice of the weights $w_{x,y}$, we  observe that the shortest path from $v_{i,0}$ to $v_{0,j}$ must first go straight up to $v_{i,j}$ and then go straight right to $v_{0,j}$.
Let $\sum c_{i,j}\cdot w_{i,j}$ be the $\pm 1$ weighted sum of the weights in the cycle. To apply Definition~\ref{def:approx_tie} and Proposition~\ref{lm:existence_epsilon_grid}, we need to show that $|c_{i,j}|\le 4$ and $\sum |c_{i,j}|>0$. For the first condition, if $w_{i,j}$ appears on the shortest path between $v_{i',0}$ to $v_{0,j'}$ and $(v_{i',0},v_{0,j'})$ is an edge in $C$, then either $i=i'$ or $j=j'$. There are at most 4 such edges in $C$, so $|c_{i,j}|\le 4$. For the second condition, we only need to show that there exists at least one $|c_{i,j}|>0$. We choose $v_{i,0}$ with the maximal $i$ appearing in $C$ and then $v_{0,j}$ which is adjacent to the $v_{i,0}$ on some occurrence of $v_{i,0}$ and has the maximal index $j$. For this edge $(a_i,b_j)$ in the cycle,  $v_{i,j}$ is passed through only once and has $c_{i,j} = \pm 1$.
\end{proof}

\begin{lemma}
For any pair of vertices $(v_{i,0},v_{0,j})~(i,j\in [0,m-1])$, if their embedded distance has error larger than $\frac{\epsilon}{2m}$, then for any query $\left(v_{i,-p},v_{-q,j}\right), (p,q\in [1,k])$ the algorithm will scan $\Omega(k)$ vertices.
\end{lemma}
\begin{proof}
In the setting of Theorem~\ref{thm:linftylb_grid}, we denote $h(s,t)=\|\pi(s)-\pi(t)\|_{\infty}$. If $\|\pi(v_{i,0})-\pi(v_{0,j})\|\le dist(v_{i,0},v_{0,j})-\frac{\epsilon}{2m}$, we can check that for query $dist(v_{i,-p},v_{-q,j})$ with $(p,q\in [1,k])$, all vertices that lie below $v_{i,0}$ or to the right of $v_{0,j}$ will be scanned by A*. We verify this for one such vertex $v_{i,-p'}(p'>0)$ here: %
\begin{align*}
&\quad dist(v_{i,-p},v_{i,-p'})+\|\pi(v_{i,-p'})-\pi(v_{-q,j})\|\\
&\le dist(v_{i,-p},v_{i,-p'})+\|\pi(v_{i,-p'})-\pi(v_{i,0})\|+\|\pi(v_{i,0})-\pi(v_{0,j})\|+\|\pi(v_{0,j})-\pi(v_{-q,j})\|\\
&\le k\epsilon_w+k\epsilon_w+dist(v_{i,0},v_{0,j})-\frac{\epsilon}{2m}+k\epsilon_w\\
&= dist(v_{i,0},v_{0,j})-\frac{\epsilon}{8m}\\
&< dist(v_{i,0},v_{0,j})<dist(v_{i,-p},v_{-q,j})
\end{align*}
The last inequality follows because we have obstacles (depicted in red) between different rows and columns in $R_r$ and $R_l$, so the shortest path from $v_{i,-p}$ to any $v_{j,0}$ must pass $v_{i,0}$. A similar statement holds for $v_{-q,j}$.
\end{proof}

\begin{lemma}
For the consistent $l_{\infty}$ norm heuristic considered by Theorem~\ref{thm:linftylb_grid}, as long as $d\le o(m)$, at least half of the pairs from $(s,t)\in \{v_{i,0}\}_{i=0}^{m-1}\times \{v_{0,j}\}_{j=0}^{m-1}$ satisfies $h(s,t)<dist(s,t)-\frac{\epsilon}{2m}$.
\end{lemma}

\begin{proof}
By the definition of $l_{\infty}$ norm, if $\|\pi(u)-\pi(v)\|_{\infty}\ge dist(u,v)-\frac{\epsilon}{2m}$, then there must exist a coordinate $i\in[d]$ where $|\pi(u)_i-\pi(v)_i|\ge dist(u,v)-\frac{\epsilon}{2m}$. We call such coordinate a ``crucial coordinate'' for the pair of vertices $(u,v)$

We claim that for one coordinate $i$, at most $2m-1$ pairs of vertices uses coordinate $i$ as their crucial coordinate. The reason is that if we construct an auxiliary graph where we assign an edge to $(u,v)$ if $(u,v)$ uses the current coordinate as their ``crucial coordinate'', the auxiliary graphs have $2m$ vertices. If there are more than $2m$ edges, at least one simple cycle exists with length $l\le 2m$. Because $l\cdot \frac{\epsilon}{2m}\le \epsilon$, such a cycle will violate Lemma~\ref{lm:auxiliary_graph}.

Therefore if there are at most $o(m)$ coordinates, there can be at most $o(m^2)$ ``crucial coordinates'', which means that at least $\frac{m^2}{2}$ pairs are embedded with error more than $\frac{\epsilon}{2m}$.
\end{proof}

\begin{proof}[Proof of Theorem \ref{thm:linftylb_grid}]
We set $m=n^{\alpha}$ and $k=n^{1-\alpha}$ then $|V|=\Theta(n)$ when $\alpha\le 0.5$.

Now we calculate the overall average complexity lower bound. Each pair $(v_{i,0},v_{0,j})$ with $h(v_{i,0},v_{0,j})<dist(v_{i,0},v_{0,j})-\frac{\epsilon}{2m}$ will produce $k^2$ pairs of queries with query complexity $\Omega(k)$, so the average query complexity is lower bounded by the following:
$$\Omega\left(\frac{\frac{m^2}{2}\cdot k^2\cdot k}{n^2}\right)=\Omega\left(\frac{n^{2\alpha}\cdot n^{2(1-\alpha)}\cdot n^{1-\alpha}}{n^2}\right)=\Omega(n^{1-\alpha})$$
\end{proof}

\subsection{A lower bound for labeling heuristics on a grid}\label{sec:info_lb_grid}
In this section we prove Theorem~\ref{thm:information_lb}, which we restate here for the sake of convenience.

\thmInfoGrid*

\paragraph{The lower bound instance.} The hard instance we use in this section has the same structure as that in Figure~\ref{fig:instance_linf_grid}, but with  a different setting of parameters.  We set the weights  as follows:
\begin{align}\label{eq:weights_info_grid}
w_{x,y}=\left\{
\begin{array}{ll}
(n-x)\cdot n^4+\delta_{x,y}\cdot n-2k & x,y\ge 0\\
1 & x \textup{ or } y<0
\end{array}\right.
\end{align}
where $\delta_{x,y}\in [-n^2,n^2]$ to be specified later.
Additionally, we define $\Delta(i,j)$ to be the sum of $\delta_{i,j}$ along the shortest path from $v_{i,0}$ to $v_{0,j}$.

\begin{lemma}
For any choice of $\delta_{x,y}$ in Equation~\ref{eq:weights_info_grid}, the shortest path of any pair $(v_{i,0},v_{0,j})$  is first going up from $v_{i,0}$ to $v_{i,j}$, and then going right to $v_{j,0}$ (recall Figure~\ref{fig:instance_linf_grid}). 
\end{lemma}

\begin{proof}
Observe that for any possible $x$ and $\delta_{x,y}$ values,  the first term in the weight $w_{x,y}$ is larger than the summation of all grids' second terms (since $x< m$ and $\delta_{x,y}\in[-n^2,n^2]$).
Hence, any shortest path must minimize the number of visited cells. 
Furthermore, the first term $(n-x)\cdot n^4$ only depends on the $x$-axis and is smaller when $x$ is large, implying that the shortest path should first go up and then go right.
\end{proof}

\begin{lemma}\label{lm:reduction_bdelta}
For any vector $x\in [0,2^b-1]^{m^2}$, there exists a choice of $\{\delta_{i,j}\}$ satisfying $|\delta_{i,j}|\le n^2$ for any $i,j$ and $\Delta(i,j)=x_{i\cdot m+j}$.
\end{lemma}

\begin{proof}
We first construct those $\{\delta_{i,j}\}$ with $i=0$ or $j=0$. We set $\delta_{0,0}=x_0$, $\delta_{i,0}=x_{i\cdot m}-x_{(i-1)\cdot m}$, and $\delta_{0,j}=x_{j}-x_{j-1}$. Then, we iterate $i,j$ in the order of $v_{1,1},\ldots,v_{1,m-1},v_{2,1},\ldots,v_{m-1,m-1}$ and set $\delta_{i,j}=\delta_{i-1,j-1}+(x_{i\cdot m+j}+x_{(i-1)\cdot m+j-1})-(x_{i\cdot m+j-1}+x_{(i-1)\cdot m+j})$. It holds that $|\delta_{i,j}|\le 2m\cdot \max_{i,j}|x_{i,j}|\le n^2$, and it can be verified that for every $i,j$, $\Delta(i,j)=x_{i\cdot m+j}$.
\end{proof}

\begin{definition}[$2m$-tuple]
We define a $2m$-tuple to be a set consisting of $2m$ vertices $(a_0,\ldots,a_{m-1},b_0,\ldots,b_{m-1})$ where
$a_i\in\{v_{i,-p}\}_{p=1}^{k}$ and similarly $b_i\in\{v_{-q,i}\}_{q=1}^{k}$. Therefore, there are $k^{2m}$ different $2m$-tuple for our graph.
\end{definition}

\begin{figure}[ht!]
\centering
\begin{tikzpicture}
\GridTuple
\end{tikzpicture}
\caption{The lower bound instance, with an example of one possible $2m$-tuple in blue.}
\end{figure}
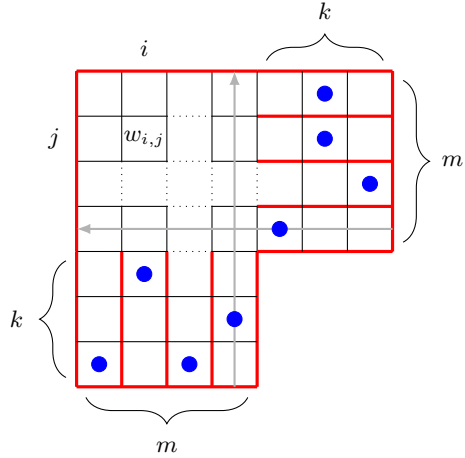

\begin{definition}[A bad pair]\label{def:bad_pair_grid}
We refer to a pair of vertices $(a_i,b_j)$ with  distortion larger than $2k$ as a \emph{bad pair}.
\end{definition}

\begin{proposition}\label{prop:bad_penality}
If a pair of vertices $(v_{i,-p},v_{-q,j})$ where $p,q> 0$ is a bad pair  as defined above, then for any shortest pair query $(v_{i,-p'}, v_{-q,j})$ ($p'> 0$)
A* will scan $v_{i,-p}$. Therefore, one bad pair contribute $\Omega(k)$ cost to the summation of all-pair shortest path queries cost.
\end{proposition}

\begin{proof}
By Fact~\ref{lm:scan_condition},  for a pair of shortest path query $(s,t)$, a vertex $u$ must be scanned by A* if $dist(s,u)+g(f(u),f(t))<dist(s,t)$. If $(v_{i,-p},v_{-q,j})$ is a ``bad" pair, it means that $g(f(v_{i,-p}),f(v_{-q,j}))< dist(v_{i,-p},v_{-q,j})-2k$. Then for any query from $v_{i,-p'}$ to $v_{-q,j}$, we have
\begin{align*}
dist(v_{i,-p'},v_{i,-p})+g(f(v_{i,-p}),f((v_{-q,j}))&< k+dist(v_{i,p},v_{-q,j})-2k<dist(v_{i,-p'},v_{-q,j})
\end{align*}
Therefore $v_{i,-p}$ will be scanned by at least $k$ queries.
\end{proof}

Next, we lower bound the number of such ``bad" pairs at $\Omega(n^2)$.

\begin{lemma}\label{lm:count_bad_pairs}
If there are $\Omega(m^2)$ bad pairs in every $2m$-tuple, then the total number of ``bad pairs" is lower bounded by $\Omega(k^2m^2)$.
\end{lemma}

The proof of this lemma is similar to Lemma~\ref{lem:bad_pairs_contribution}, so we omit it here.

\begin{lemma}\label{lm:memorizeweight_grid}
Consider the graph described above and the A* algorithm using a consistent $L=o(m)$ labeling heuristic. If for any choice of $\{\delta_{ij}\}$ there exists a $2m$-tuple $(a_0,\ldots,a_{m-1},b_0,\ldots,b_{m-1})$, at most $o(m^2)$ pairs of the induced all-pair embedding distances $\{g(f(a_i),f(b_j))\}_{i,j\in[m]}$ have large distortion, $g(f(a_i),f(b_j))<dist(a_i,b_j)-2k$, then there is an one-way protocol for indexing problem using $o(nb)$ bits and succeeds with probability at least $1-o(1)$.
\end{lemma}

\begin{proof}
Suppose for every choice of $\{\delta_{i,j}\}$, there exists at least one $2m$-tuple with their induced all-pair embedding distances satisfying that for at least $1-o(1)$ fraction of the pairs, $dist(a_i,b_j)-2k\le g(f(a_i),f(b_j))\le dist(a_i,b_j)$. Then we construct the following protocol for the indexing problem. We let $n=m^2$ and we think of the index that Bob receives as a tuple $(i,j)$ where $i,j\in [m]$.

Alice: given an input $x\in [0,2^b-1]^{m^2}$, invoke Lemma~\ref{lm:reduction_bdelta} to produce $\{\delta_{i,j}\}$ s.t. $\forall (i,j), \Delta_{i,j}=x_{i\cdot m+j}$. Enumerate over all the $2m$-tuples to find the one such that  for at least $1-o(1)$ fraction of its pairs, the  embedded distance  distortion is smaller than $2k$.  Note that smaller than $2k$ distortion means that $dist(v_{i,0},v_{0,j})-2k\le g(f(a_i),f(b_j))\le dist(v_{i,0},v_{0,j})+2k$ so that $\left\lceil\frac{g(f(a_i),f(b_j))}{n}\right\rceil-(n-i)\cdot n^3=\Delta_{i,j}$. Alice sends to Bob $\{f(a_i)\}$ and $\{f(b_j)\}$, which has bit length  $o(m^2b)$.

Bob: Receive $\{f(a_i)\}$ and $\{f(b_j)\}$ from Alice, and use them to construct $\Delta'_{i,j}=\left\lceil\frac{g(f(a_i),f(b_j))}{n}\right\rceil-(n-i)\cdot n^3$, where at least $1-o(1)$ fraction of them satisfy $\Delta'_{i,j}=x_{i,j}$. Given Bob's query $(i,j)$, Bob answers the corresponding term $\Delta'_{i,j}$.

Therefore, the protocol above can solve the indexing problem with probability greater than
$1-o(1)$ and communication complexity $o(nb)$ for any input.
\end{proof}

\begin{proof}[Proof of Theorem~\ref{thm:information_lb}]
We set $m=n^{\alpha}$ and $k=n^{1-\alpha}$ with ($\alpha\le 0.5$) implying that the size of the graph is still $\Theta(n)$. We prove this theorem by contradiction. If there exists a consistent labeling heuristic $h$ with length $L=o(n^{\alpha})$ such that for any choice of $\{\delta_{ij}\}$, A* with heuristic $h$ scans $o(n^{1-\alpha})$ vertices. Then by Proposition~\ref{prop:bad_penality} and Lemma~\ref{lm:count_bad_pairs}, for each choice of $\{\delta_{ij}\}$ we can find a $2m$-tuple, there are at most $o(m^2)$ bad-pairs. Finally, by Lemma~\ref{lm:memorizeweight_grid}, we can get an algorithm for indexing problem contradicting Theorem~\ref{thm:result_general_index}.
\end{proof}

\section{Proof of Proposition~\ref{prop:prob_unique}}\label{appen:prob_unique_proof}

\begin{proof}
We define a boolean variable $tie(s,t)$ which equals 1 if and only if there exists two equal length shortest path from $s$ to $t$, $N(s)$ to denote the set of neighbors of vertex $s$, $w(u,v)$ to be the edge weight for edge $(u,v)$.

\begin{align*}
&\quad Pr[\cup_{s,t\in V}\mb{I}[tie(s,t)]]\\
&\le \sum_{s,t}\sum_{u,v\in N(s) \land u\neq v}Pr[w(s,u)+dist(u,t)=w(s,v)+dist(v,t)]\\
&\le n^4Pr[w(s,u)+dist(u,t)=w(s,v)+dist(v,t)]
\end{align*}
In the following, we prove that for any fixed $t,u,v$ where $u,v\in N(s)$ and $u\neq v$, the expression above $Pr[w(s,u)+dist(u,t)=w(s,v)+dist(v,t)]$ is bounded by $n^{-5}$. We use $d_{s,u}(\cdot,\cdot)$ to be the shortest distance using all edges in $E\setminus (s,u)$, and $d(\cdot,\cdot)$ to be the shortest distance using all edges in $E$ (an abbreviation for $dist(\cdot,\cdot)$). Then, the probability of the existence of non-unique shortest path can be upper bounded as the following:
\begin{align}
&\quad Pr[w(s,u)+d(u,t)=w(s,v)+d(v,t)]\notag\\
&=Pr[w(s,u)=w(s,v)+d(v,t)-d(u,t)]\notag\\
&=Pr[d(v,t)=d_{s,u}(v,t) \land d(u,t)=d_{s,u}(u,t) \land w(s,u)=w(s,v)+d(v,t)-d(u,t)]\notag\\
&\ +\underbrace{Pr[(d(v,t)\neq d_{s,u}(v,t) \lor d(u,t)\neq d_{s,u}(u,t)) \land w(s,u)=w(s,v)+d(v,t)-d(u,t)]}_{\textup{this case is impossible}}\label{line:whetherdchange}\\
&\le Pr[w(s,u)=w(s,v)+d_{s,u}(v,t)-d_{s,u}(u,t)]\label{line:csuequal}\\
&\le \max_{C}Pr[w(s,u)=C]\notag\\
&\le n^{-5}\notag
\end{align}
On line~\ref{line:whetherdchange}, we expand the probability basing on whether the distance function changes after adding edge $(s,u)$. We can observe that as long as the addition of edge $(s,u)$ influences the shortest path $P(v,t)$ or $P(u,t)$, the later equality cannot hold. In line~\ref{line:csuequal}, $w(s,u)$ is independent of all the other quantities. Therefore, we prove that the probability of violation of Assumption~\ref{asm:stronger_unique_shortest_path} is upper bounded by $\frac{1}{n}$.

Because our edge weight are drawn from a discrete distribution, unique shortest path is enough to show the existence of a constant margin as required in Assumption~\ref{asm:stronger_unique_shortest_path}.
\end{proof}

\end{document}